\documentclass[11pt]{article}
\usepackage{template}
\usepackage{dutchcal}
\usepackage[normalem]{ulem}

\newcommand{\trotter}{\textup{trot}}
\newcommand{\best}{\eta^\star}
\newcommand{\trotterstep}[1]{\tau^{\trotter}_{#1}}
\newcommand{\F}{\textup{F}}
\newcommand{\MP}{\textup{MP}}
\newcommand{\PP}{\textup{PP}}
\newcommand{\Trunc}{\textup{Trunc}}

\newcommand{\hc}{\mathcal{h}}
\newcommand{\kk}{k}
\newcommand{\op}{\textup{op}}
\newcommand{\tmax}{t_\textup{max}}

\newcommand{\ceil}[1]{\left\lceil #1 \right\rceil}
\renewcommand{\gamma}{\upgamma}
\newcommand{\dt}{\delta t}
\renewcommand{\r}{\mathsf{r}}
\newcommand{\s}{\mathsf{s}}

\title{Fast convergence of Majorana propagation\\
for weakly interacting fermions
}

\usepackage{authblk}
\author[1]{Giorgio~Facelli}
\author[1]{Hamza~Fawzi}
\author[2]{Omar~Fawzi}
\affil[1]{Department of Applied Mathematics and Theoretical Physics, University of Cambridge, United Kingdom}
\affil[2]{Inria, ENS Lyon, UCBL, LIP, F-69342 Lyon Cedex 07, France}

\date{\today}

\begin{document}

\maketitle

\begin{abstract}
Simulating the time dynamics of an observable under Hamiltonian evolution is one of the most promising candidates for quantum advantage as we do not expect efficient classical algorithms for this problem except in restricted settings. Here, we introduce such a setting by showing that Majorana Propagation, a simple algorithm combining Trotter steps and truncations, efficiently finds a low-degree approximation of the time-evolved observable as soon as such an approximation exists. This provides the first provable guarantee about Majorana Propagation for Hamiltonian evolution. As an application of this result, we prove that Majorana Propagation can efficiently simulate the time dynamics of any sparse quartic Hamiltonian up to time $\tmax(u)$ depending on the interaction strength $u$. For a time horizon $t \leq \tmax(u)$, the runtime of the algorithm is $N^{O(\log(t/\eps))}$ where $N$ is the number of Majorana modes and $\eps$ is the error measured in the normalized Frobenius norm. Importantly, in the limit of small $u$, $\tmax(u)$ goes to $+\infty$, formalizing the intuition that the algorithm is accurate at all times when the Hamiltonian is quadratic.
\end{abstract}

\section{Introduction}

Simulating the time evolution of a quantum system is one of the most promising applications of quantum computers~\cite{feynman1982simulating,lloyd1996universal} and today's intermediate-scale quantum computers are already achieving remarkable performance~\cite{kim2023evidence,alam2025fermionic,alam2025programmable} with limited error correction mechanisms. To benchmark such quantum algorithms, it is essential to develop \emph{classical} algorithms for quantum simulation and understand their limits. 
Simulating the time dynamics of a local observable under Hamiltonian evolution is a BQP-complete problem so we do not expect efficient classical algorithms for this problem except in some restricted settings such as short-time dynamics~\cite{wild2023classical}.

We focus in this work on a system of interacting fermions described by a quartic Hamiltonian $H$ in $N$ Majorana modes $\gamma_1,\ldots,\gamma_N$ (assuming $N$ is even). Given an observable $A$, our objective is to approximate the Heisenberg evolution
\be
\label{eq:intro1timeevolution}
A(t) := e^{\i H t} A e^{-\i H t}.
\ee
Note that any observable $A$ admits a unique expansion as $A = \sum_{X \subset [N]} a_X \gamma_X$ where $a_X \in \CC$ and $\gamma_X$ is a Hermitian Majorana string of length $|X|$ (see Section~\ref{sec:preliminaries} for a precise definition). The \emph{degree} of $A$ is defined as the largest $|X|$ such that $a_X \neq 0$.

\subsection{Contributions}

\paragraph{Majorana Propagation for real-time quantum dynamics} We analyze a simple algorithm to compute an approximation of the time-evolved observable $A(t)$. The idea of the algorithm is to perform a discretization of the time evolution via Trotterization and apply a \emph{propagation} algorithm in the basis of Majorana operators. Recently proposed propagation algorithms for circuit simulation proceed by tracking the evolution of an observable in a well-chosen basis and performing truncations to ensure the observable remains low-complexity and thus to keep the algorithm efficient. Such algorithms have been studied for simulating circuits with low Wigner negativity~\cite{pashayan2015estimating}, near-Clifford circuits~\cite{bennink2017unbiased,rall2019simulation,beguvsic2025simulating} and for noisy/random circuits~\cite{aharonov2023polynomial,schuster2025polynomial,fontana2025classical,angrisani2025classically,angrisani2026simulating}. In very recent work, the same methodology was applied in the fermionic setting for circuits~\cite{miller2025simulation} and Hamiltonian systems~\cite{danna2025majorana}.

Specifically, the Majorana Propagation (MP) algorithm we analyze in this paper produces an approximation $A_{\MP}(t)$ of $A(t)$
of the form\footnote{assuming $t/\dt$ is an integer for simplicity}
\be
\label{eq:MP-output}
A_{\MP}^{\ell,\dt}(t) := (\Trunc_{\ell} \circ \trotterstep{\dt})^{t/\dt} (A) \,,
\ee
where $\trotterstep{\dt}$ is a Trotter approximation of the Heisenberg evolution $A\mapsto e^{\i H \dt} A e^{-\i H \dt}$, and $\Trunc_{\ell}:\cA \to \cA$ is some truncation map on the space of observables $\cA$ that depends on a threshold value $\ell$. Several truncation strategies have been discussed in previous works such as degree-based truncation (see e.g.,~\cite{rudolph2025pauli} for a discussion), or more sophisticated rules based on the number of unpaired Majoranas~\cite{danna2025majorana}. The truncation we will consider in this work is based solely on the degree, where for $A = \sum_{X \subset [N]} a_{X} \gamma_X$, $\Trunc_{\ell}(A)$ simply discards all Majorana strings of degree larger than $\ell$ in the Majorana expansion of $A$ (see Section \ref{sec:preliminaries} for more details). A sketch of the different steps involved in the MP algorithm is presented in Fig. \ref{fig:conceptual-figure}(a).

\begin{figure}
    \centering
    \includegraphics[width=\textwidth]{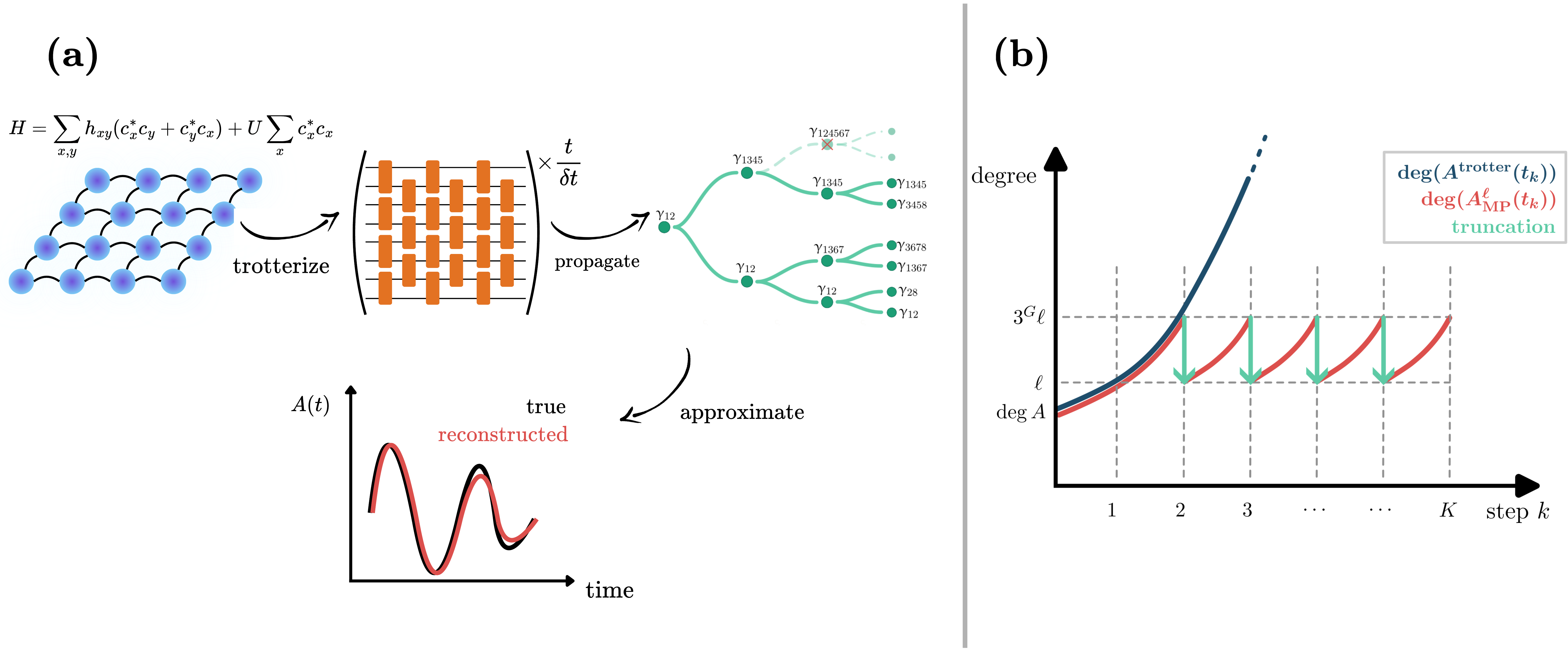}
    \caption{\textbf{(a)} A sketch showing the main steps involved in the algorithm: given a system of interest, the first step is to trotterize the dynamics generated by the Hamiltonian. Then, the trotterized circuit is propagated through the observable, and the Majorana strings of degree larger than some threshold are discarded. After propagation, the time-evolved observable can then be approximated. \textbf{(b)} A diagram showing the degree as a function of the trotter step. In the standard trotterized evolution (blue line) the degree grows very quickly and is not bounded, while for the output of the MP algorithm (red line) the degree is ensured to be at most $3^G\ell$ and after each trotter step is truncated to degree $\ell$ (see section \ref{sec:MP-descr} for more details).}
    \label{fig:conceptual-figure}
\end{figure}
For any time $t \geq 0$, the MP algorithm with parameter $\ell$ produces a degree-$\ell$ approximation of $A(t)$. Observe that the time-evolved observable $A(t) = e^{\i H t} A e^{-\i H t}$ can have maximal degree $N$ as soon as $t > 0$. In order to measure the accuracy of the MP algorithm, it is then natural to introduce the following quantity, which quantifies the error of the \emph{best} degree-$\ell$ approximant of $A(t)$ for a given norm $\|\cdot\|$:\footnote{The maximization over $s \in [0,t]$ in the definition of $\best(H,A,t,\ell)$ is to guarantee that $\best$ is nondecreasing in $t$.}
\be
\label{eq:bestell}
\best(H,A,t,\ell) = \max_{s \in [0,t]} \min_{P \in \cA\,:\,\deg(P)\leq \ell} \|P - A(s)\|.
\ee
It is clear, by definition of the MP algorithm, that the error incurred by the MP algorithm over $[0,t]$ is at least as large as $\best(H,A,t,\ell)$, since the latter gives the error of the \emph{best} degree-$\ell$ approximant. One may ask whether $\best(H,A,t,\ell)$, and the best approximant $P$ in \eqref{eq:bestell} can be computed efficiently solely using the input $H,A$ and $t$. Our main result shows that the error incurred by the output of the MP algorithm cannot be much larger than the error of the best approximant when $\|\cdot\|$ is the (normalized) Frobenius norm, denoted $\|\cdot\|_{\F}$.
\begin{theorem}[Analysis of Majorana Propagation for time dynamics -- Informal version of Thm.~\ref{thm:majorana-prop}]
\label{thm:informalMPintro}
    Let $H$ be a $\Delta$-sparse quartic fermionic Hamiltonian in $N$ Majorana modes (see~\eqref{eq:vertexbounded} for the definition of $\Delta$-sparse) and $A$ an arbitrary observable with $\norm{A}_\F=1$. Let $t \geq 0$ be a time horizon.
    Then the Majorana Propagation algorithm with time step $\dt$ and truncation degree $\ell \geq \deg(A)$ has time and memory complexity at most $(t/\dt) \cdot N^{O(\ell)}$ and produces a degree-$\ell$ observable $A^{\dt, \ell}_{\MP}(t)$ that satisfies
\be
\label{eq:boundthm2}
    \|A^{\dt, \ell}_{\MP}(t) - A(t)\|_\F  = 
    O\Big(\ell^2 t \dt + (t/\dt)\best_F(H,A,t,\ell) \Big),
\ee
where $A(t)$ is the time-evolved observable \eqref{eq:intro1timeevolution}, and $\best_{\F}$ is the quantity \eqref{eq:bestell} when $\|\cdot \| = \|\cdot \|_{\F}$.
\end{theorem}
Some comments are in order:
\begin{itemize}
\item The only assumption on the Hamiltonian, beyond the fact that it is quartic, is that it is $\Delta$-sparse, which means that any Majorana mode $\gamma_i$ for $i \in [N]$ is present in at most $\Delta$ terms of the Hamiltonian.
Note that this goes beyond any assumption on geometric locality, hence Lieb-Robinson-type bounds are not readily available in this case.
\item The right-hand side of \eqref{eq:boundthm2} is composed of two terms: the first term is due to the discretization of the time dynamics in intervals of equal length $\dt$, whereas the second term is due to the truncation map which restricts the time-evolved operator to be of degree at most $\ell$. A simple calculation shows that the right-hand side of \eqref{eq:boundthm2} is minimized at
\[
\dt^\star = \Theta\left(\frac{\sqrt{\best_\F(H,A,t,\ell)}}{\ell \Delta } \right) \,,
\]
and the error \eqref{eq:boundthm2} reduces to:
\[
\|A^{\dt^\star, \ell}_{\MP}(t)-A(t)\|_{\F} = O\left(\ell t \sqrt{\best_\F(H,A,t,\ell)}\right).
\]
From the expression above, we see that with an adequate step size, the error of MP is essentially governed by the error of the best degree-$\ell$ approximation of $A(t)$.
\item The MP algorithm considered in the theorem above uses a specific Trotter decomposition strategy whereby the Hamiltonian terms are partitioned into $G \propto \Delta$ groups, i.e., $H = H_1 + \dots + H_G$ such that in each $H_i$, the Majorana strings are pairwise disjoint  (such a partitioning is guaranteed to exist by the $\Delta$-sparsity assumption), and the small-time evolution $e^{\i H \dt} A e^{-\i H \dt}$ is approximated by
\[
e^{\i H \dt} A e^{-\i H \dt} \approx e^{\i H_G \dt} \dots e^{\i H_1 \dt} A e^{-\i H_1 \dt} \dots e^{-\i H_G \dt} =: \trotterstep{\dt}(A).
\]
This particular partitioning is used in the analysis of the algorithm to control the growth of $\deg \trotterstep{\dt}(A)$. A visualization of this is illustrated in a diagram in Fig. \ref{fig:conceptual-figure}(b).
\item The Frobenius norm between observables can be interpreted in terms of the mean squared difference between the corresponding expectation values for a state randomly chosen from a suitable \emph{low-average} ensemble; see Section~\ref{sec:norms} for details.
\end{itemize}

\paragraph{Application to weakly interacting Hamiltonians} As an application of Theorem \ref{thm:informalMPintro}, we consider the case of weakly interacting Hamiltonians of the form
\[
H = H_0 + uV
\]
where $H_0$ is quadratic, $V$ is quartic, and $u \leq 1$. It is known that if $u=0$, then $\deg(A(t)) = \deg(A)$ for all $t \geq 0$, or equivalently that
\be
\label{eq:bestquadratic}
\best(H,A,t,\ell) = 0 \quad \forall t \geq 0 \text{ and } \ell \geq \deg(A).
\ee
Using perturbation-theory techniques, we show that for any observable $A$, there is a $\tmax(u) > 0$ such that $A(t) = e^{\i (H_0 + uV) t} A e^{-\i (H_0 + u V) t}$ can be very well approximated by a low-degree Majorana polynomial up to time $\tmax(u)$. More precisely, we prove the following.
\begin{theorem}[Low-degree approximability of $A(t)$ for weakly interacting Hamiltonians]
\label{thm:intro-weak-interaction}
Let $H = H(u) = H_0 + uV$ be a $\Delta$-sparse fermionic Hamiltonian where $H_0$ and $V$ are quadratic and quartic interactions, respectively. Let $A$ be an arbitrary observable. Then, for all $t < \tmax(u) := \log(e/u) / (8 e^2 \Delta (\deg A + 2))$,
\be
\label{eq:weakinteractiontruncationerrorfrob}
\best_\F(H(u),A,t,\ell) \leq \dfrac{\left(t/\tmax(u) \right)^{(\ell-\deg(A))/2}}{1-t/\tmax(u)} \|A\|_{\F}.
\ee
\end{theorem}
An important feature of the theorem above is that $\tmax(u) \to +\infty$ as $u\to0$, and thus the Theorem recovers the well-known fact \eqref{eq:bestquadratic} for quadratic Hamiltonians. Also observe that if, e.g., $t \leq \tmax(u)/2$ then $\best_F(H(u),A,t,\ell)$ decays exponentially in $\ell-\deg(A)$; in other words, one can get an $\eps$-approximation (in the Frobenius norm) of $A(t)$ with a polynomial of degree $\deg(A) + O(\log(1/\eps))$.

The proof of Theorem \ref{thm:intro-weak-interaction} is constructive, i.e., for any $t$ and $\ell$, we exhibit a degree-$\ell$ Majorana polynomial $P$ whose error $\|P-A(t)\|_{\F}$ is upper bounded by the right-hand side of \eqref{eq:weakinteractiontruncationerrorfrob}. However, even though $P$ is explicit, it is defined in terms of infinite series and it is not clear whether one can compute it efficiently. Nonetheless, by using Theorem \ref{thm:informalMPintro} we can show that since there \emph{exists} a good low-degree polynomial approximant for $A(t)$, then Majorana Propagation will be able to find such an approximant, efficiently.

It is interesting to compare the bound of Theorem \ref{thm:intro-weak-interaction} with a generic application of Lieb-Robinson bounds. Given $A=\sum_{X\subseteq [N]} a_X\gamma_X$, define $\norm{\va}=\sum_{X\subseteq [N]} \abs{a_X}$ as the $1$-norm of the Majorana coefficients.

Then, it is not difficult to show the following upper bound on $\best_\op$ via Lieb-Robinson bounds:
\be
\label{eq:lrboundetaop}
    \best_\op(H,A,t,\ell) \leq \frac{1}{2}\norm{\va}\deg A \exp \left[ v_\text{LR}t - \omega^{-1}\left(\dfrac{\ell}{\deg A}\right)\right]
\ee
where
\be
\omega(r) := \max_{x \in [N]} |\{y \in [N] : d_H(x,y) \leq r\}|
\ee
is the volume of the largest ball of radius $r$ in the interaction graph associated to $H$ (see Appendix \ref{appendix:lieb-robinson} for the detailed definitions). For an interaction graph corresponding to a $D$-dimensional grid, we have $\omega(r) \simeq r^D$ and so $\omega^{-1}(\ell) \simeq \ell^{1/D}$. However for general interaction graphs of degree $\Delta$, one can have $\omega(r)$ that grows expoentially with $r$, in which case the bound \eqref{eq:lrboundetaop} is exponentially worse than \eqref{eq:weakinteractiontruncationerrorfrob}. Furthermore, the general bound \eqref{eq:lrboundetaop} does not account for the case of quadratic Hamiltonians, whose Heisenberg evolution is fully simulable for initially local operators. This is instead naturally integrated in Theorem \ref{thm:intro-weak-interaction} via the quantity $\tmax(u)$.

\begin{corollary}[Efficient simulation of weakly interacting Hamiltonians with Majorana Propagation]
\label{corollary:MPweaklyinteracting}
Let $H = H_0 + u V$ be a $\Delta$-sparse fermionic Hamiltonian on $N$ Majorana modes where $H_0$ and $V$ are quadratic and quartic, respectively. Let $A \in \cA$ be an arbitrary observable. Then, for any $0 < \eps < 1$ and any $t \leq \log(e/u) / (16 e^2 \Delta (\deg A + 2))$, the Majorana Propagation algorithm computes an approximation $\tilde{A} \in \cA$ in complexity $N^{O(\log((1+t)/\eps))}$ such that $\|A(t) - \tilde{A}\|_{\F} \leq \eps \|A\|_{\F}$. The approximant $\tilde{A}$ is represented as a degree-$\ell$ Majorana polynomial where $\ell = \deg A + O(\log((1+t)/\eps))$.
\end{corollary}
We note that for Gibbs sampling, efficient quantum algorithms~\cite{vsmid2025polynomial} and more recently classical algorithms~\cite{chen2025convergence} were designed for weakly interacting fermions.
For time dynamics, the closest result to Corollary~\ref{corollary:MPweaklyinteracting} we are aware of is a general algorithm with runtime doubly exponential in $t$ for approximating the expectation value of the time-evolved observable for product states~\cite{wild2023classical}. Note that~\cite{wild2023classical} considers spin systems instead of fermionic systems.

As a final remark, we note that all of the results above apply in a similar way to spin systems, thereby extending our results to \emph{Pauli} propagation as well. Theorem \ref{thm:informalMPintro} straightforwardly generalizes to the Pauli case. Our second result, Theorem \ref{thm:intro-weak-interaction}, also applies by considering single-site Pauli strings in $H_0$ and multi-site strings in $V$. We detail this in Appendix \ref{appendix:pauli-propagation}. Our focus on the fermionic setting is mostly motivated by closer ties to physically relevant scenarios where a free-fermion system is perturbed by some complex interaction which induces highly non-trivial dynamics. This is further emphasized by recent experimental results in quantum simulation, which also look at fermionic dynamics as promising candidate application of quantum computers as well as a testbed where rich physical phenomena can be investigated \cite{alam2025fermionic,alam2025programmable}.

\subsection{Overview of techniques}
\label{sec:ingredients}

The main technical ingredients in this paper are new Frobenius norm inequalities which depend on \emph{degree} rather than system size. We highlight here two such inequalities.

\begin{itemize}
\item \emph{Norm of commutator:} Observe that if $H$ is a Hamiltonian on $N$ modes, then\footnote{Indeed, note that the linear map $[H,\cdot] : \cA \to \cA$ can be written as $H \otimes 1 - 1\otimes H$ and so has eigenvalues $\lambda_i(H) - \lambda_j(H)$. We are using the fact that if $\cL:\cA\to\cA$ is a linear operator which is self-adjoint with respect to the trace inner product, then $\max_{A \in \cA} \|\cL(A)\|_{\F} / \|A\|_{\F}$ is equal to the spectral radius of $\cL$.}
\begin{equation}
\label{eq:worstcasecommutatorbound}
\max_{A \in \cA} \frac{\|[H,A]\|_{\F}}{\|A\|_{\F}} = \lambda_{\max}(H) - \lambda_{\min}(H) \gtrsim \Omega(N).
\end{equation}
We prove in Theorem \ref{thm:commutatorfrobeniusbound} that if we restrict the maximization above to operators $A$ which are of degree at most $d$, then the value is upper bounded by a quantity that only depends on $d$, and not on $N$. More precisely we show that if $H$ is a quartic Hamiltonian which is $\Delta$-sparse, then
\be
\label{eq:intro-comutator-norm}
\max_{\substack{A \in \cA\\ \deg(A) \leq d}} \frac{\|[H,A]\|_{\F}}{\|A\|_{\F}} \leq 2 \Delta \sqrt{d(d+2)}.
\ee
\item \emph{Trotter error:} Assume now that $H$ is a quartic Hamiltonian which is split into $G$ groups, namely
\[
H = H_1 + \dots + H_G,
\]
and define, respectively
\be
\begin{aligned}
\tau_t(A) &= e^{\i H t} A e^{-\i H t}\\
\trotterstep{t}(A) &= e^{\i H_G t} \dots e^{\i H_1 t} A e^{-\i H_1 t} \dots e^{-\i H_G t}.
\end{aligned}
\ee
An important question in the simulation of quantum dynamics is to bound the error $\|\tau_t(A) - \trotterstep{t}(A)\| / \|A\|$, see e.g., \cite{childs2021theory}. It is easy to verify that $\tau_t(A) - \trotterstep{t}(A) \approx C(H_1,\ldots,H_G) t^2 + O(t^3)$ as $t\to0$, more precisely
\be
\label{eq:intro-limit-trotter-error}
\lim_{t\to0} \frac{\tau_t(A) - \trotterstep{t}(A)}{t^2} = \frac{1}{2} \sum_{1\leq i < j \leq G} [[H_i,H_j],A].
\ee
As such, we have,
\[
\max_{A \in \cA} \lim_{t\to0} \frac{1}{t^2} \frac{\|\tau_t(A) - \trotterstep{t}(A)\|_{\F}}{\|A\|_{\F}} = \frac{1}{2} \sum_{1\leq i < j \leq G} \lambda_{\max}([H_i,H_j]) - \lambda_{\min}([H_i,H_j]).
\]
For standard examples of Hamiltonians, it is easy to see that the right-hand side scales like $\Omega(N)$. In this paper we show that we can obtain a system-size independent bound on $\|\tau_t(A) - \trotterstep{t}(A)\|_{\F}$ provided that $A$ is restricted to be of low degree. Namely if $H$ is quartic and $\Delta$-sparse, then we prove in Theorem \ref{thm:trotter-error} that
\be
\label{eq:intro-trotter-error}
\max_{\substack{A \in \cA\\ \deg(A) \leq d}} \frac{\|\tau_t(A) - \trotterstep{t}(A)\|_{\F}}{\|A\|_{\F}} \leq C \cdot (d \Delta G)^2 \cdot t^2
\ee
for some absolute constant $C > 0$.
\end{itemize}

\section{Preliminaries}
\label{sec:preliminaries}

\paragraph{Majorana algebra} For even $N$, a system with $N/2$ fermionic modes is described by a Fock space $(\mathbb{C}^{2})^{\otimes N/2}$ with an occupation number basis $\{\ket{n_1, \dots, n_{N/2}}\}_{n_1, \dots, n_{N/2} \in \{0,1\}}$. The algebra $\cA$ of operators acting on this Fock space is generated by the creation and annihilation operator $c_j^{*}, c_j$ for $j \in [N/2]$, where $\ket{n_1, \dots, n_{N/2}} = (c_1^{*})^{n_1} \cdots (c^*_{N/2})^{n_{N/2}} \ket{0,\dots,0}$. A state is a positive linear functional $\rho : \cA \to \CC$ satisfying $\rho = 1$. A pure state takes the form $\rho(A) = \bra{\psi} A \ket{\psi}$ for some unit vector $\ket{\psi} \in (\mathbb{C}^{2})^{\otimes N/2}$, which we also call a state. In this paper, we describe $\cA$ with Majorana operators, defined as
\begin{equation}
\begin{cases}
\gamma_{2j} &= c_j^* + c_j\\
\gamma_{2j+1} &= \i (c_j^* -  c_j).
\end{cases}
\end{equation}
The Majorana operators satisfy the canonical anticommutation relations (CAR)
\[
\{\gamma_i,\gamma_j\} = 2\delta_{ij}\,,\quad i,j\in \{1,\ldots,N\} \,.
\]
Any observable $A \in \cA$ has a unique expansion as
\[
A = \sum_{X \subset [N]} a_X \gamma_X
\]
where, for each $X \subset [N]$,
\be
\label{eq:majorana-string-def}
    \gamma_X := \i^{r_X} \prod_{x\in X} \gamma_x \,,
\ee
where the product is ordered such that $x_i < x_j$ for $i<j$, and where we set $r_X =0$ for $\abs{X}= 0,1 \; \mod 4$ and $r_X = 1$ otherwise (in particular $r_X=|X|(|X|-1)/2 \; \mod 2$). As such, $\gamma_X$ is Hermitian for all $X$. We use the convention that $\gamma_{\emptyset} = 1$, and we let $\cM$ be the collection of all $2^N$ Majorana monomials. The \emph{degree} of an observable $A$ is defined as
\[
\deg(A) = \max_{X \subset [N]} \{ |X| : a_X \neq 0 \}.
\]
We recall the following facts about Majorana monomials.
\begin{proposition}
\label{prop:majoranabasics}
For all $S,T\subset [N]$, we have the following properties
\begin{enumerate}
\item Multiplication of two Majorana monomials
\[
    \gamma_S \gamma_T = \i^{(\abs{S}\abs{T}+\abs{S\cup T})\mod 2} (-1)^{|\{(t,s) \in T \times S : t < s\}|} \gamma_{S\triangle T}.
\]
\item $\{\gamma_S,\gamma_T\} = 0 \iff [\gamma_S,\gamma_T] = 2\gamma_S \gamma_T \iff \abs{S}\abs{T}-\abs{S \cap T} \text{ is odd}$
\item $[\gamma_S,\gamma_T] = 0 \iff \{\gamma_S,\gamma_T\} = 2\gamma_S \gamma_T \iff \abs{S}\abs{T}-\abs{S \cap T} \text{ is even}$
\item For any $\theta \in \RR$
\begin{equation}
\label{eq:formula_conjugation_gamma}
e^{\i \frac{\theta}{2} \gamma_T} \gamma_S e^{-\i \frac{\theta}{2} \gamma_T} = \begin{cases}
    \gamma_S & \text{ if } [\gamma_S,\gamma_T] = 0\\
    \cos(\theta) \gamma_S + \i \sin(\theta) \gamma_T \gamma_S & \text{ else.}
\end{cases}
\end{equation}
\end{enumerate}
\end{proposition}

We will also need the following simple corollary.
\begin{corollary}
For any $S \subset [N]$, $S \neq \emptyset$ such that $|S|$ is even, and $A \in \cA$,
\begin{equation}
\label{eq:degreecommutatormonomial}
\deg([\gamma_S,A]) \leq \deg(A) + |S|-2.
\end{equation}
\end{corollary}
\begin{proof}
It suffices to show that for any $T$, $\deg([\gamma_S,\gamma_T]) \leq |T| + |S| - 2$. If $|S \cap T|$ is even then $[\gamma_S,\gamma_T] = 0$ and there is nothing to do. If $|S\cap T|$ is odd then $\deg([\gamma_S,\gamma_T]) = |S\triangle T| = |S|+|T|-2|S\cap T| \leq |S|+|T|-2$ as desired.
\end{proof}

\subsection{Norms on $\cA$}
\label{sec:norms}

To quantify the approximation of observables, one can use several norms on $\cA$. For an observable $A$ and an estimator $\widehat{A}$ (assumed to be Hermitian), the operator norm $\| \cdot \|_{\op}$ can be interpreted as the difference of expectation values on the worst-case state
\[
\| A - \widehat{A} \|_{\op} = \sup \{ |\rho(A) - \rho(\widehat{A})| : \rho \text{ state}\} = \sup \{ |\bra{\psi} A \ket{\psi} - \bra{\psi} \widehat{A} \ket{\psi}| : \ket{\psi} \text{ state}\}.
\]
In this paper, we work with the normalized Frobenius norm defined, for $A \in \cA$, as
\[
\|A\|_{\F} = \sqrt{\tr(A^* A)},
\]
where $\tr$ is the normalized trace, i.e., $\tr(1) = 1$. Note that
\be
\| \gamma_{X} \|_{\F} = 1 \text{ and } \tr(\gamma_X) = 0 \;\; \forall X \neq \emptyset.
\ee
If the observable $A$ has an expansion
\[
A = \sum_{X \subset [N]} a_X \gamma_X
\]
then
\be
\label{eq:frobA}
\|A\|_{\F} = \left(\sum_{X \subset [N]} |a_X|^2\right)^{1/2}.
\ee
While there are multiple interpretations on the distance in Frobenius norm between observable $A$ and $\widehat{A}$1, we mention here one that fits particularly well within our context, of which a variant of it is also employed in \cite{schuster2025polynomial}. Consider some ensemble of states $\cE :=\{\rho_k\}_k$, which is assumed to be of \emph{low-average}, i.e. for some $c\in O(1)$ we have
\[
    \norm{ \dfrac{1}{\abs{\cE}}\sum_{k} \rho_k }_\op = \dfrac{c}{2^{N/2}} \,.
\]
In words, this condition is equivalent to the top eigenvalue of the mixture of states in $\cE$ being not much larger than the eigenvalue corresponding to the maximally mixed state (a sort of "randomness" measure). This condition is satisfied for example if the $\rho_k$'s are pure states forming an orthonormal basis, in which case $c=1$.
One can then show that, for any $X\in\cA$ we have
\[
\begin{aligned}
    \E_{\rho\sim\cE} \bigg[ \abs{\Tr(\rho X)}^2 \bigg] &= \E_{\rho \sim \cE} \bigg[ \Tr(\rho X)^* \Tr(\rho X) \bigg] \\
    &\leq \E_{\rho \sim \cE} \bigg[ \Tr[X^*X\rho] \bigg] \\
    &= \Tr[X^*X\E_{\rho \sim \cE}[\rho] ] \\
    &\leq c \norm{X}^2_{\F}
\end{aligned}
\]
where $\Tr$ here denotes the standard trace. The first inequality is a straightforward consequence of the fact that, for any state $\rho$ we have
\[
\Tr[(X-\langle X\rangle_\rho)^*(X-\langle X\rangle_\rho)\rho]\geq 0
\]
with $\langle X\rangle_\rho = \Tr[X\rho]$.  In our case we would set $X=A-\hat{A}$ in order to probe the accuracy of the estimator. Obtaining a bound on the Frobenius norm would then imply an upper bound on the root mean squared error of the expectation value of $A-\hat{A}$ over the ensemble of states $\cE$. Further interpretations of the Frobenius norm of operators are also discussed in~\cite[Chapter 8]{chen2023speed}.

\subsection{Hamiltonian}
\label{sec:hamiltonian}
We define a Hamiltonian $H \in \cA$ with quartic interactions as a Hermitian operator of the form:
\be\label{eq:hamiltonian}
    H := \sum_{X\in\cX} h_X \quad \text{where}\quad h_X = \hc_X \gamma_X, \;\; \hc_X \in \CC, |\hc_X|\leq 1 \, .
\ee
The Hamiltonian is assumed to only have terms of degree two or four, as this is the case for the majority of physical fermionic Hamiltonians. For this reason, we will write
\be
\label{eq:hamiltonianterms24}
\cX = \cX_2 \cup \cX_4
\ee
where
\be
\begin{aligned}
    \cX_2 := \{X \subset [N] : \abs{X}=2\,\text{and}\,h_X \neq 0\} \,, \\
    \cX_4 := \{X \subset [N] : \abs{X}=4\,\text{and}\,h_X \neq 0\} \,.
\end{aligned}
\ee
Throughout this work, we will assume that the Hamiltonian is $\Delta$-sparse, i.e.,
\be
\label{eq:vertexbounded}
\max_{i\in[N]}\abs{\{X \in \cX : i\in X \text{ and } h_X \neq 0\}} \leq \Delta\,.
\ee

The sparsity assumption implies a number of useful and easy-to-prove consequences that we list here.

\begin{proposition}
\label{prop:bounded-degree-property}
For any Majorana monomial $A \in \cM$, the number of terms $h_X$ of the Hamiltonian such that $[h_X,A] \neq 0$ is bounded above by $\Delta \deg A$, i.e.,
\begin{equation}
\label{eq:num-nnz-comm}
|\{X \in \cX : [h_X, A] \neq 0\}| \leq \Delta \deg A \,.
\end{equation}
\end{proposition}
\begin{proof}
First observe that given $p \in \{1,\ldots,N\}$, there are at most $\Delta$ terms of the Hamiltonian that have a non-zero commutator with the degree-one monomial $\gamma_p$. This follows directly from the assumption \eqref{eq:vertexbounded}, and proves \eqref{eq:num-nnz-comm} when $A$ has degree 1.
For an arbitrary monomial $A = \gamma_{p_1} \dots \gamma_{p_d}$ where $d = \deg A$, observe that for any $h \in \cA$, we have
\[
[h,A] = \sum_{j=1}^d \gamma_{p_1}\dots \gamma_{p_{j-1}}[h,\gamma_{p_j}]\gamma_{p_{j+1}}\dots \gamma_{p_{d}}.
\]
Hence
\[
|\{X : [h_X,A] \neq 0\}| \leq \sum_{j=1}^d |\{X : [h_X,\gamma_{p_j}] \neq 0\}| \leq \Delta d
\]
as desired.
\end{proof}

When analyzing the complexity of MP, it will be important to control the growth of the degree of the observable after each Trotter step. To achieve this, we next show how one can color the terms in the Hamitonian such that any two terms with the same color have disjoint support. With this decomposition, given any observable $A$ it will then be easier to control by how much the degree of $A$ grows after having applied one Trotter step.
\begin{proposition}
\label{prop:partitionHamiltonian}
There is an efficient algorithm that can compute a partition of the terms $\cX$ of $H$ in $G \leq 4\Delta$ groups
\begin{equation}
\label{eq:Xpartition}
\cX = \cX^1 \cup \dots \cup \cX^G
\end{equation}
such that for any $i \in \{1,\ldots,G\}$ and any $X,Y \in \cX^i$, $X\cap Y = \emptyset$.
\end{proposition}
\begin{proof}
Consider the graph $\cG=(\cX,E)$ where each node corresponds to a term of the Hamiltonian, and where an edge exists between two terms $X,Y \in \cX$ iff 
$X\cap Y \neq \emptyset$.
Finding a partition such as the one in Eq.~\eqref{eq:Xpartition} amounts to finding a coloring of $\cG$. Denoting by $\deg_{\cG}(X)$ the number of edges at node $X\in\cX$, we have that
\[
\deg_\cG(X) \leq \sum_{x\in X} \abs{\{Y\in\cX : x \in Y\text{ and } Y\neq X\}} \leq \sum_{x\in X} (\Delta  - 1) \leq 4(\Delta-1)\,,
\]
since by assumption on the set $\cX$, $\abs{X}\leq 4$. A standard result~\cite{diestel2025graph} is that a simple greedy algorithm can color the graph using $\deg_\cG(X) + 1 \leq 4\Delta$ colors.
\end{proof}

\section{Majorana Propagation algorithm}
\label{sec:MP}

We first describe in Section \ref{sec:MP-descr} precisely the MP algorithm that we consider in this work, and state the main result in Theorem~\ref{thm:majorana-prop} concerning its convergence. Before delving in the proof, in Section~\ref{sec:MP-ingredients} we discuss two important technical results needed to analyze the convergence and complexity of the algorithm. The proof of Theorem~\ref{thm:majorana-prop} is then delayed to Section~\ref{sec:MP-thm-proof}.

\subsection{Description of the algorithm}
\label{sec:MP-descr}

\paragraph{Trotter dynamics} 
The MP algorithm we consider in this section starts by decomposing the Hamiltonian $H$ into $G$ terms
\[
    H = \sum_{g=1}^G H^g \,, \quad H^g = \sum_{X\in\cX^g} h_X
\]
such that for each $g \in \{1,\ldots,G\}$, all the monomial terms in $H^g$ have disjoint support. Proposition \ref{prop:partitionHamiltonian} shows that one can efficiently find such a decomposition with  $G \leq 4 \Delta$. For any $A\in\cA$, $g\in[G]$ and $\dt$, let
\be
\label{eq:tauHg}
\tau^{g}_{\dt}(A) = e^{\i \dt H^g} A e^{-\i \dt H^g}\,,
\ee
and let $\trotterstep{\dt}$ be a Trotter approximation of the unitary evolution operator $\tau_{\dt}$ obtained as the composition of $\{\tau^g_{\dt}\}_{g=1}^G$
\be
\label{eq:deftautrotter}
\trotterstep{\dt}(A) = (\tau_{\dt}^G \circ \dots \circ \tau_{\dt}^1)(A)\,.
\ee
Note that since the individual terms inside each $H^g$ pairwise commute, we have
\[
    \tau^g_{\dt}(A) = \prod_{X\in\cX^g} e^{\i \dt h_X} A \prod_{X\in\cX^g} e^{-\i \dt h_X}
\]
and therefore Eq.~\eqref{eq:deftautrotter} is equivalent to first-order Trotter, where each single Hamiltonian term gives rise to the adjoint action of a single unitary, and the ordering of the unitaries is dictated by the decomposition. This particular ordering allows us to control the growth of the degree during a Trotter step. Indeed, the following short Proposition shows that for each $g \in [G]$, $\deg \tau^g_{\dt}(A) \leq 3 \deg(A)$, which implies that $\deg(\trotterstep{\dt}(A)) \leq 3^G \deg(A)$.

\begin{proposition}
\label{prop:commutingterms}
Let $S \subset [N]$ and let $\gamma_{X_1},\ldots,\gamma_{X_k}$ be a set of degree-two and degree-four Majorana monomials whose supports are pairwise disjoint (hence they pairwise commute). Define, for any choice of real numbers $\{\theta_i\}$, the following operator
\[
O = \left(\prod_{i=1}^k e^{\i\theta_i \gamma_{X_i}}\right) \gamma_S \left(\prod_{i=1}^k e^{-\i\theta_i \gamma_{X_i}}\right).
\]
Then $O$ is a Majorana polynomial that is a sum of at most $2^{|S|}$ monomials and satisfies $\deg(O) \leq 3|S|$. Moreover, it can be computed in time $O(k + 2^{|S|})$.
\end{proposition}
\begin{proof}
We can assume without loss of generality that $X_i \cap S \neq \emptyset$ for all $i=1,\ldots,k$ (indeed, if $X_i \cap S = \emptyset$ then $[\gamma_{X_i},\gamma_S] = 0$ and the corresponding Majorana rotation has no effect). Since the $X_i$ are pairwise disjoint, this implies that $k \leq |S|$. Furthermore, we know from Proposition \ref{prop:majoranabasics} that for any $X$ with $|X| \leq 4$
\[
\deg\left(e^{\i \theta \gamma_X} \gamma_S e^{-\i \theta \gamma_X}\right) \leq \max(|S|,|S\triangle X|) \leq |S| + 2.
\]
Hence since $O$ is the result of applying $k$ operations of this form to $\gamma_S$, we get
\[
\deg(O) \leq |S| + 2k \leq |S|+2|S| = 3|S|
\]
as desired. To compute $O$, we first loop through the sets $X_i$ and only keep the ones for which $X_i \cap S \neq \emptyset$, the corresponding runtime is $O(k)$. To compute the product, we use the formula in~\eqref{eq:formula_conjugation_gamma} sequentially and each time the total number of Majorana monomials gets multiplied by at most $2$. The total number of times we use Eq.~\eqref{eq:formula_conjugation_gamma} is at most $1 + 2 + 2^2 + \dots + 2^{|S|} = O(2^{|S|})$.
\end{proof}

\paragraph{Degree truncation} In the MP algorithm we truncate the observable to degree at most $\ell$ after each Trotter evolution. More precisely, consider the truncation map $\Trunc_\ell:\cA \to \cA$ defined on any observable $A=\sum_{Y\subset [N]} a_X \gamma_X$ as
\[
    \Trunc_\ell(A) =   \sum_{\substack{X\subset [N]\\\ |X| \leq \ell}} a_X \gamma_X.
\]
The truncation map satisfies the following important properties.
\begin{proposition}
\label{prop:truncationmap}
For any observable $A \in \cA$, and any $\ell \in \NN$, the following holds:
\begin{itemize}
    \item[(i)] $\|\Trunc_{\ell}(A)\|_{\F} \leq \|A\|_{\F}$ 
    \item[(ii)] $\min_{P \in \cA, \deg P \leq \ell} \|P - A\|_{\F} = \|\Trunc_{\ell}(A) - A\|_{\F} = (\sum_{X \subset [N], |X| > \ell} |a_X|^2)^{1/2}$.
\end{itemize}
\end{proposition}
\begin{proof}
This is a direct consequence of \eqref{eq:frobA}.
\end{proof}

\paragraph{Majorana Propagation algorithm} Using this notation, we can express the MP algorithm in simple terms. Choose a time step $\dt$ and then let, assuming $t/\dt$ is an integer,
\be
\label{eq:AellP}
A^{\dt,\ell}_{\MP}(t) =  \left(\Trunc_{\ell} \circ \trotterstep{\dt}\right)^{t/\dt}(A)\,.
\ee

\begin{remark}
If $t/\dt$ is not an integer, we adapt the definition \eqref{eq:AellP} to use a smaller step size in the last iteration, namely
\be
\label{eq:AellP2}
A^{\dt,\ell}_{\MP}(t) =  \Trunc_{\ell} \circ \trotterstep{t-\lfloor t/\dt \rfloor \dt} \circ \left(\Trunc_{\ell} \circ \trotterstep{\dt}\right)^{\lfloor t/\dt \rfloor}(A).
\ee
\end{remark}

\begin{algorithm}[ht]
\label{alg:MP}
\caption{Majorana Propagation algorithm}
{\bf Input}: Quartic Majorana Hamiltonian $H$ with sparsity $\Delta$, Observable $A$, Horizon $t$\\
{\bf Parameters}: $\dt > 0$, $\ell \geq \deg(A)$\\
{\bf Output}: Estimate of $A(t) = e^{\i H t} A e^{-\i H t}$\\
1. Decompose the Hamiltonian $H$ into $G \leq 4\Delta$ as $H = H^1 + \dots +H^G$ such that for any $g \in \{1,\ldots,G\}$, all the monomial terms in $H^g$ have disjoint supports (see Proposition \ref{prop:partitionHamiltonian}) \\
2. Define $\trotterstep{\dt}$ according to Eq.~\eqref{eq:deftautrotter}\\
3. Return $A^{\dt, \ell}_{\MP}(t)$ as defined in \eqref{eq:AellP2}
\end{algorithm}

We are now ready to state our main result concerning the analysis of MP.

\begin{theorem}[Analysis of Majorana propagation]
\label{thm:majorana-prop}
Let $A \in \cA$ be an arbitrary observable, and let $A^{\dt, \ell}_{\MP}(t)$ be the result of applying the MP algorithm with parameters $\dt$ and $\ell \geq \deg(A)$ for a time horizon $t \geq 0$ (see Algorithm \ref{alg:MP}). Then:
\begin{itemize}
\item[\textup{(i)}] The operator $A^{\dt, \ell}_{\MP}(t)$ can be computed in time complexity at most $\frac{t}{\dt} \cdot N^{O(\ell)}$.
\item[\textup{(ii)}] We have the bound
\begin{equation}
\label{eq:bounderrppell}
        \|A^{\dt,\ell}_{\MP}(t)-A(t)\|_{\F} \leq 34 t\cdot \dt \cdot \Delta^2(\ell+2)^2 \|A\|_{\F} + \left\lceil \dfrac{t}{\dt}\right \rceil \best_{\F}(H,A,t,\ell)\,,
\end{equation}
where $A(t) = e^{\i H t} A e^{-\i H t}$ and $\best_{\F}(H,A,t,\ell)$ is the quantity \eqref{eq:bestell} with $\|\cdot\| = \|\cdot\|_{\F}$.
\end{itemize}
\end{theorem}

\subsection{Ingredients}
\label{sec:MP-ingredients}
We start by proving two results, introduced in Section \ref{sec:ingredients}, that will play an important role in the proof of Theorem \ref{thm:majorana-prop}.

\subsubsection{Bound on the commutator's Frobenius norm} 

\begin{theorem}
\label{thm:commutatorfrobeniusbound}
For any observable $A \in \cA$, we have
\[
\|[H,A]\|_{\F} \leq 2 \Delta \sqrt{\deg A(\deg A+2)} \|A\|_{\F}.
\]
\end{theorem}
\begin{proof}
Write $A = \sum_{Y \subset [N]} a_Y \gamma_Y$, and $H = \sum_{X \in \cX} \hc_X \gamma_X$, where $|\hc_X| \leq 1$ for all $X \in \cX$. From Proposition \ref{prop:majoranabasics} we know that for $X \in \cX$ and arbitrary $Y \subset [N]$, $[\gamma_X,\gamma_Y] = 0 \text{ if } |X \cap Y|$ even, and otherwise $[\gamma_X,\gamma_Y] = 2 d_{X,Y} \gamma_{X\Delta Y}$ for some unit complex number $d_{X,Y} \in \CC$. Thus we have
\[
\begin{aligned}
[H,A] = 2 \sum_{\substack{X \in \cX, Y \subset [N]\\ |X\cap Y|\text{ odd}}} \hc_X a_Y d_{X,Y} \gamma_{X \triangle Y} = 2 \sum_{Z \subset [N]} \left(\sum_{\substack{X \in \cX,Y \subset [N]\\ |X\cap Y|\text{ odd},\, X \triangle Y = Z}} \hc_X a_Y d_{X,Y}\right) \gamma_{Z}.
\end{aligned}
\]
We can then write
\begin{equation}
    \label{eq:H,Afrobnd1}
\begin{aligned}
\|[H,A]\|^2_{\F} &= 4 \sum_{Z \subset [N]} \left|\sum_{(X,Y) \in \cP(Z)} \hc_X a_Y d_{X,Y}\right|^2
\end{aligned}
\end{equation}
where, for a fixed $Z \subset [N]$, we let
\[
\cP(Z) = \{(X,Y) : X \in \cX, Y \subset [N] \text{ s.t. } X \triangle Y = Z \text{ and } |X \cap Y| \text{ odd}\}.
\]
By the Cauchy-Schwarz inequality, we have for any $Z$, and using the fact that $|\hc_X d_{X,Y}| \leq 1$
\begin{equation}
    \label{eq:CS}
\left|\sum_{(X,Y) \in \cP(Z)} \hc_X a_Y d_{X,Y}\right|^2 \leq |\cP(Z)| \sum_{(X,Y) \in \cP(Z)} |a_Y|^2.
\end{equation}
Note that
\[
\begin{aligned}
|\cP(Z)| &= |\{X \in \cX : |X\cap (X \triangle Z)| \text{ odd}\}|\\
&\overset{(a)}{=} |\{X \in \cX : |X \cap Z| \text{ odd}\}|\\
&= |\{X \in \cX : [\gamma_X, \gamma_Z] \neq 0\}|\\
&\overset{(b)}{\leq} \Delta |Z|
\end{aligned}
\]
where in (a) we used the fact that $|X\cap (X\triangle Z)| = |X| - |X\cap Z|$ and that $|X|$ is even, and in (b) we used Proposition \ref{prop:bounded-degree-property}.
Plugging \eqref{eq:CS} in \eqref{eq:H,Afrobnd1} we get
\[
\begin{aligned}
\|[H,A]\|_{\F}^2 &\leq 4\sum_{Z \subset [N]} \Delta |Z| \sum_{(X,Y) \in \cP(Z)} |a_Y|^2\\
&= 4 \Delta \sum_{Y \subset [N]} |a_Y|^2 \sum_{X \in \cX, |X\cap Y| \text{ odd}} |X \triangle Y|\\
&\leq 4 \Delta^2 \sum_{Y \subset [N]} |a_Y|^2 |Y|(|Y|+2)\\
&\leq 4 \Delta^2 (\deg A) (\deg A+2) \|A\|^2_{\F}
\end{aligned}
\]
where we used the fact that for $X \in \cX$, and $a_Y \neq 0$, $|X\triangle Y| \leq |Y|+2 \leq \deg A + 2$, and that  $|\{X \in \cX : |X \cap Y| \text{ odd}\}| \leq \Delta |Y| \leq \Delta \deg A$ using Proposition \ref{prop:bounded-degree-property}.
This concludes the proof.
\end{proof}

\subsubsection{Trotter error bound}
\label{app:trotter}

The second result deals with the Trotter error bound for low-degree observables.
\begin{theorem}
\label{thm:trotter-error}
Let $\tau_t(A) = e^{\i H t} A e^{-\i H t}$, and let $\trotterstep{t}(A)$ be the Trotter approximation of $\tau_t$ as defined in \eqref{eq:deftautrotter}, for a given decomposition $H = H^1 + \dots + H^G$ of the Hamiltonian. Then for any $A \in \cA$ with $\deg A = d$, we have
\begin{equation}
\label{eq:alphatrotterbnd11}
\|\tau_t(A) - \trotterstep{t}(A)\|_{\F} \leq 2(\Delta \cdot t\cdot (d+2))^2\left(G^2+ 1\right).
\end{equation}
\end{theorem}
\begin{remark}
The bound above holds for an arbitrary decomposition $H = H^1 + \dots + H^G$, i.e., it does not use the fact that the monomial terms inside each $H^i$ have disjoint supports. If the latter holds, then one can prove the following improved bound:
\begin{equation}
\label{eq:alphatrotterbnd22}
\|\tau_t(A) - \trotterstep{t}(A)\|_{\F} \leq 2( t\cdot (d+2))^2\left(G^2+ \Delta^2\right).
\end{equation}
\end{remark}
\begin{proof}[Proof of Theorem \ref{thm:trotter-error}]
We first prove the following general lemma, an application of Taylor's theorem with integral remainder.

\begin{lemma}
\label{lemma:taylor-remainder}
For any finite-dimensional super-operators $\cL_1,\ldots,\cL_G$,
\be
\label{eq:taylor-remainder-trotter}
\begin{aligned}
e^{t\cL_G} \cdots e^{t\cL_1} &= \id + t (\cL_G+\cdots+\cL_1)\\
& \quad + \sum_{i=1}^G \int_{0}^{t} \di \tau e^{t\cL_G} \cdots e^{t\cL_{i+1}} e^{\tau \cL_i} \cL_i \left( (t-\tau) \cL_i + t \sum_{j=1}^{i-1} \cL_{j} \right).
\end{aligned}
\ee
\end{lemma}
\begin{proof}
The proof is by induction on $G$. The identity is clearly true for $G=1$ using the Taylor's theorem with integral remainder. Assuming the identity holds for some $G$, we can write:
\[
\begin{aligned}
e^{t\cL_{G+1}}e^{t\cL_G} \cdots e^{t\cL_1} &= e^{t\cL_{G+1}} + t e^{t\cL_{G+1}}(\cL_G+\cdots+\cL_1) \\
&+\sum_{i=1}^G \int_0^t \di\tau e^{t\cL_{G+1}}\cdots e^{t\cL_{i+1}} e^{\tau \cL_i} \cL_i ((t-\tau) \cL_i + t\cL_{i-1} + \dots + t \cL_1).
\end{aligned}
\]
We expand $e^{t\cL_{G+1}}$ with Taylor's theorem to $p$-th order; in the first term for $p=1$, and $p=0$ for the second one to get:
\[
\begin{aligned}
e^{t\cL_{G+1}} \cdots e^{t\cL_1} &= \id + t\cL_{G+1} +\int_0^t\di\tau (t-\tau) e^{\tau \cL_{G+1}}\cL_{G+1}^2 \\
& \quad \quad +  t\left(\id+\int_{0}^{t} d\tau e^{\tau \cL_{G+1}} \cL_{G+1}\right)(\cL_G+\cdots+\cL_1)\\
&\quad \quad +\sum_{i=1}^G \int_0^t \di\tau e^{t\cL_{G+1}}\cdots e^{t\cL_{i+1}} e^{\tau \cL_i} \cL_i ((t-\tau) \cL_i + t\cL_{i-1} + \dots + t \cL_1)
\end{aligned}
\]
which, after rearranging, gives the formula of Eq.~\eqref{eq:taylor-remainder-trotter} for $G+1$. Hence this holds for general $G$.
\end{proof}

For any $i = 1,\ldots,G$, let $\cL_i = \i[H^i,\cdot]$, and let $\cL = \i[H,\cdot]$ so that $\cL = \cL_1 + \dots + \cL_G$. Note that $e^{t \cL}(A) = e^{\i H t} A e^{-\i Ht} = \tau_t(A)$ and $(e^{t\cL_G}\cdots e^{t\cL_1})(A) = \trotterstep{t}(A)$. Using the expansion \eqref{eq:taylor-remainder-trotter} of $e^{t\cL_G}\cdots e^{t\cL_1}$, and using
\[
e^{t \cL} = \id + t\cL + \int_{0}^{t} \di \tau (t-\tau) e^{\tau \cL} \cL^2
\]
we have, by the triangle inequality:
\begin{equation}
\label{eq:trotterdiff123}
\begin{aligned}
\norm{(e^{t\cL_G}\cdots e^{t\cL_1}- e^{t\cL})(A)}_{\F} &\leq \sum_{i=1}^{G} \int_{0}^{t}  \di \tau \norm{e^{t\cL_G} \cdots e^{t\cL_{i+1}} e^{\tau \cL_i} \cL_i \left( (t-\tau) \cL_i + t\cL_{i-1} + \dots + t\cL_1 \right)(A)}_{\F}\\
& \qquad + \int_{0}^{t} \di \tau (t-\tau) \norm{e^{\tau \cL} \cL^2(A)}_{\F}.
\end{aligned}
\end{equation}
Since $e^{s \cL_i}$ is a conjugation by a unitary for $s \in \RR$, we have, $\norm{e^{s \cL_i}(B)}_{\F} = \norm{B}_{\F}$ for any operator $B \in \cA$. Hence the integrand of the first term in \eqref{eq:trotterdiff123} can be bounded above by
\[
\begin{aligned}
\norm{(t-\tau) \cL_i^2(A) + t\cL_i (\cL_{i-1} + \dots + \cL_1)(A)}_{\F}
&\leq (t-\tau) \norm{\cL_i^2(A)}_{\F} + t \sum_{j=1}^{i-1} \norm{\cL_i \cL_j(A)}_{\F},
\end{aligned}
\]
and the integrand of the second term in \eqref{eq:trotterdiff123} is bounded by $(t-\tau) \norm{\cL^2(A)}_{\F}$. Putting things together, we get
\be
\label{eq:trotter-frobenius-sum}
\norm{(e^{t\cL_G}\cdots e^{t\cL_1}- e^{t\cL})(A)}_{\F} \leq \frac{t^2}{2} \left( \norm{\cL^2(A)}_{\F} + \sum_{i=1}^{G} \norm{\cL_i^2(A)}_{\F} + 2 \sum_{1\leq j < i \leq G} \norm{\cL_i \cL_j(A)}_{\F} \right).
\ee
Using Theorem \ref{thm:commutatorfrobeniusbound}, we know that
\[
\|\cL(A)\|_{\F} \leq 2 \Delta \sqrt{\deg A(\deg A+2)} \|A\|_{\F}.
\]
By applying the inequality again, and using the fact that $\deg \cL(A) \leq \deg A + 2$, we get
\[
\|\cL^2(A)\|_{\F} \leq 4 \Delta^2 (\deg A+2) \sqrt{\deg A(\deg A+4)} \|A\|_{\F} \leq 4 \Delta^2 (\deg A+2)^2 \|A\|_{\F}.
\]
The same bound applies to the quantities $\norm{\cL_i^2(A)}_{\F}$ and $\norm{\cL_i\cL_j(A)}_{\F}$. Hence we get at the end
\[
\norm{(e^{t\cL_G}\cdots e^{t\cL_1}- e^{t\cL})(A)}_{\F} \leq 2 t^2 \Delta^2 (\deg A + 2)^2 (1+G^2).
\]
This proves \eqref{eq:alphatrotterbnd11}.

To prove \eqref{eq:alphatrotterbnd22}, it suffices to note that if for a given $g \in [G]$, $X\cap Y = \emptyset$ for all $X,Y\in\cX^g$ then $H^g$ is $1$-sparse, i.e., 
\[
    \max_{x\in [N]} |\{X\in\cX^g: x \in X \text{ and }h_X\neq 0 \}| = 1.
\]
Hence, it follows from Theorem~\ref{thm:commutatorfrobeniusbound} that
\[
    \norm{\cL_i(A)}_\F \leq 2\sqrt{\deg A (\deg A + 2)}\norm{A}_\F,
\]
and so $\norm{\cL_i \cL_j(A)}_{\F} \leq 4 (\deg A+2)^2 \norm{A}_{\F}$ for all $i,j$.
Using this fact in \eqref{eq:trotter-frobenius-sum} gives \eqref{eq:alphatrotterbnd22}.
\end{proof}

\subsection{Proof of Theorem \ref{thm:majorana-prop}}
\label{sec:MP-thm-proof}
We are now ready to prove the main theorem on the analysis of Majorana Propagation.

\begin{proof}[Proof of Theorem~\ref{thm:majorana-prop}]

Let $t_0 = 0 < t_1 < \dots < t_K = t$ be the time steps of the algorithm, i.e., $t_k = k \cdot \dt$ for $k=0,\ldots,\lfloor t/\dt \rfloor$ and $K = \lceil t/\dt \rceil$. For convenience, we omit the superscript $\dt$ and write $A^{\ell}_{\MP}(t)$ instead of $A^{\ell,\dt}_{\MP}(t)$.
\begin{itemize}
\item[(i)] We first analyze the runtime for computing $A_{\MP}^\ell(t)$.
At each time step $t_k$, we have a Majorana polynomial $A_{\MP}^\ell(t_k)$ of degree at most $\ell$ and thus is a sum of at most $N^{\ell}$ monomials. Using Proposition~\ref{prop:commutingterms}, computing $\tau^1_{\dt}(A_{\MP}^\ell(t_k))$ can be done in time $O(N^{\ell}(N+2^{\ell})) = O(N^{2\ell})$ and the resulting polynomial has degree at most $3\ell$. For the second step, applying Proposition~\ref{prop:commutingterms}, we see that computing $\tau^2_{\dt} \circ \tau^1_{\dt}(A_{\MP}^\ell(t_k))$ can be done in time $O(N^{3\ell}(N+2^{3\ell})) = O(N^{6\ell})$. As a result, computing $\trotterstep{\dt}(A_{\MP}^\ell(t_k))$ can be done with a total runtime of $O(N^{2\ell} + N^{6\ell} + \dots + N^{2 \cdot 3^G \ell}) = N^{O(\ell)}$ as $G \leq 4\Delta$ for quartic $\Delta$-sparse Hamiltonians.

Even though we do not pursue this further, we remark that
it is possible to analyze more closely the evolution of the number of monomials through the algorithm. Using Proposition~\ref{prop:commutingterms}, we have that $\trotterstep{\dt}(A)$ has degree at most $3^G \deg(A)$ and if $A$ has at most $s$ terms, then $\trotterstep{\dt}(A)$ can be written as a sum of at most $s \cdot 2^{\deg(A)} \cdot 2^{3\deg(A)} \cdots 2^{3^G\deg(A)} \leq s \cdot 2^{3^{G+1}\deg(A)}$ monomials. 
In addition, the truncation step cannot increase the number of monomials. As such, each step of the algorithm multiplies that number of terms by at most $2^{3^{G+1}\ell}$.

\item[(ii)] We now analyze the error. Let us first consider the case where $t/\dt$ is an integer. By the definition of $A_{\MP}^{\ell}(t)$ and the triangle inequality, we have that
        \be
        \label{eq:erroranaly1}
        \begin{aligned}
            \norm{A_{\MP}^\ell(t) - A(t)}_{\F} &= \norm{\Trunc_\ell\circ\trotterstep{\dt} (A_{\MP}^\ell(t-\dt)) - A(t)}_{\F} \\
            &\leq \norm{\Trunc_\ell\left(\trotterstep{\dt} (A_{\MP}^\ell(t-\dt)) - A(t)\right)}_{\F} + \norm{\Trunc_\ell(A(t)) - A(t)}_{\F}.
        \end{aligned}
        \ee
Using the fact that $\Trunc_{\ell}$ is nonexpansive for the Frobenius norm, and
letting
\be
\tau_{\dt}(\cdot) = e^{\i H \dt} (\cdot) e^{-\i H \dt}
\ee
we have, continuing from \eqref{eq:erroranaly1},
\begin{equation}
\label{eq:mainproof234}
\begin{aligned}
\|A_{\MP}^\ell(t) - A(t)\|_{\F} &\leq \|\trotterstep{\dt} (A_{\MP}^\ell(t-\dt)) - A(t)\|_\F +  \|\Trunc_\ell(A(t)) - A(t)\|_\F \\
&\leq \|\trotterstep{\dt} (A_{\MP}^\ell(t-\dt)) - \tau_{\dt} (A_{\MP}^\ell(t-\dt))\|_{\F} \\
&\qquad + \|\tau_{\dt} (A_{\MP}^\ell(t-\dt))- \tau_{\dt}(A(t-\dt))\|_{\F} \\
&\qquad + \|\Trunc_\ell(A(t)) - A(t)\|_{\F}.
\end{aligned}
\end{equation}

We can use \eqref{eq:alphatrotterbnd22} from Theorem~\ref{thm:trotter-error} to bound the first term as
\[
\|\trotterstep{\dt} (A_{\MP}^\ell(t-\dt)) - \tau_{\dt} (A_{\MP}^\ell(t-\dt))\|_{\F} \leq \kappa \cdot (\dt)^2 \|A_{\MP}^\ell(t-\dt)\|_{\F} \overset{(a)}{\leq} \kappa \cdot (\dt)^2 \|A\|_{\F}
\]
where we let
\begin{equation}
\label{eq:defkappa}
\kappa := 2(\ell+2)^2(G^2+\Delta^2) \overset{(b)}{\leq} 34 \Delta^2 (\ell+2)^2
\end{equation}
and where we used in $(a)$ that $\|A_{\MP}^{\ell}(t-\dt)\|_{\F} \leq \|A\|_{\F}$ (a consequence of the nonexpansiveness of $\Trunc_{\ell}$ for the Frobenius norm, and the unitarity of $\trotterstep{\dt}$), and in $(b)$ we used the fact that $G \leq 4 \Delta$ (Proposition \ref{prop:partitionHamiltonian}).

Continuing from \eqref{eq:mainproof234}, we get
\be
    \|A_{\MP}^\ell(t) - A(t)\|_\F \leq
\kappa \cdot (\dt)^2 \|A\|_{\F} + \|A_{\MP}^\ell(t-\dt) - A(t-\dt)\|_{\F} + \|\Trunc_\ell(A(t)) - A(t)\|_{\F}\,. \label{eq:recursive-ineq}
\ee
Remark that the second term in Eq.~\eqref{eq:recursive-ineq}  is identical to the LHS but at time $t-\dt$. We can then iterate the inequalities to get
\be
\label{eq:mainprooffinal}
\begin{aligned}
    \|A^\ell_{\MP}(t)-A(t)\|_{\F} &\leq \frac{t}{\dt} \kappa\cdot (\dt)^2 \|A\|_{\F} + \sum_{k=1}^{t/\dt} \|\Trunc_\ell(A(k \cdot \dt)) - A(k \cdot \dt)\|_{\F}\\
    &\leq t \cdot \kappa \cdot \dt \|A\|_{\F} + \frac{t}{\dt} \best_{\F}(t, \ell)\,,
\end{aligned}
\ee
where we used the fact that $\|\Trunc_\ell(A(k \cdot \dt))- A(k \cdot \dt)\|_{\F}\leq\best_{\F}(t, \ell)$. Plugging the bound \eqref{eq:defkappa} for $\kappa$ completes the proof of Theorem \ref{thm:majorana-prop} in the case where $t/\dt$ is an integer.

The case where $t/\dt$ is not an integer is identical, except that the first term in \eqref{eq:mainprooffinal} will be of the form $\kappa \cdot (s^2 + \frac{t_{K-1}}{\dt} (\dt)^2) \|A\|_{\F}$ where $t_{K-1} = \lfloor t/\dt \rfloor \dt$ and $s = t-t_{K-1} \in (0,\dt)$. Since $s \leq \dt$ and $t = t_{K-1}+s$ then we still have $\kappa \cdot (s^2 + t_{K-1} \dt) \leq \kappa \cdot t \cdot \dt$. In addition, the second term in \eqref{eq:mainprooffinal} will be of the form $\ceil{\frac{t}{\dt}} \best_{\F}(t, \ell)$.
\end{itemize}
\end{proof}

\section{Weakly interacting fermions}
\label{sec:weakinteractinglimit}

In this section, we assume our Hamiltonian $H$ has the form
\be
\label{eq:weakInteractionHamiltonian}
    H := H_0+ uV = \sum_{X\in\cX_2} h_X + u \sum_{X\in\cX_4} h_X.
\ee
where we recall that $\cX_2,\cX_4$ are degree-two and degree-four sets of interactions and $u \in [0,1]$. 
Our goal is to prove Theorem \ref{thm:intro-weak-interaction}, which we recall here for convenience:
\begin{theorem}
\label{thm:weak-interaction}
Assume $H$ is a $\Delta$-sparse Hamiltonian of the form \eqref{eq:weakInteractionHamiltonian}. Let $A$ be an arbitrary observable. Then for all $t < \tmax(u) := \log(e/u) / (8 e^2 \Delta (\deg A + 2))$ and any $\ell \geq \deg A$ such that
\be
\label{eq:weakinteractiontruncationerrorfrob2}
\best_{\F}(H,A,t,\ell) \leq \frac{\left(t/\tmax(u) \right)^{(\ell-\deg(A))/2}}{1-t/\tmax(u)} \|A\|_{\F}
\ee
\end{theorem}
The proof involves two natural steps: First, in Section~\ref{sec:dysonseries} we derive a series expansion of $A(t)$ in powers of $u$ that is reminiscent of Dyson series, where the $k$'th term of the series has degree at most $\deg(A)+2k$. Secondly, we show in Section~\ref{sec:approximatedegree} that truncating the series derived will incur the error stated in Equation \eqref{eq:weakinteractiontruncationerrorfrob2}.

\subsection{Expressing $A(t)$ as a series in $u$}
\label{sec:dysonseries}

It is well-known that $A(t) = e^{\i t H} A e^{-\i t H}$ has the following multi-commutator series expansion
\begin{equation}
    \label{eq:nestedcomm-expansion}
    A(t) = \sum_{p=0}^{\infty} \frac{(\i t)^p}{p!} [H,[H,\dots [H,A]]]\,.
\end{equation}
By using the fact that $H = H_0 + u V$, and rearranging the terms according to the power of $u$, one can show the following result.
\begin{lemma}
Given a Hamiltonian $H = H_0 + uV$ and an arbitrary observable $A\in\cA$, the time-evolved observable $A(t)$ has the following series expansion
\be
\label{eq:alternativedyson}
    A(t) = \sum_{n=0}^\infty u^n\sum_{m_0,\dots,m_n=0}^\infty \dfrac{t^{n+\sum_i m_i}}{(n+\sum_i m_i)!}\cF^{m_n}\cG\cF^{m_{n-1}}\cdots \cF^{m_1} \cG \cF^{m_0} A
\ee
where $\cF(\cdot) := \i [H_0,\cdot]$ and $\cG(\cdot) := \i[V,\cdot]$.
\end{lemma}
\begin{proof}
For any integer $p \geq 0$, we can expand the $p$-th order nested commutators in \eqref{eq:nestedcomm-expansion} and rearrange the terms according to the power of $u$ as follows:
\begin{equation}
\label{eq:expandcommrearrange}
\begin{aligned}
\i^p [H,[H,\dots [H,A]]]
&= \i^p [H_0+uV,[H_0+uV,\dots [H_0+uV,A]]] \\
&= \sum_{n=0}^p u^n \sum_{\substack{m_0, \dots, m_n \\ m_0+m_1+\cdots+m_{n} + n = p}} \cF^{m_n} \cG \cF^{m_{n-1}} \cdots \cF^{m_1} \cG \cF^{m_0} A.
\end{aligned}
\end{equation}
To obtain the desired equality \eqref{eq:alternativedyson}, we simply need to plug \eqref{eq:expandcommrearrange} into \eqref{eq:nestedcomm-expansion} and swap the summations over $p$ and $n$. To justify this last step, define the sequence $\alpha_{p,n}$ as 
\[
\alpha_{p,n} = \frac{(\i t)^p}{p!} u^n \sum_{\substack{m_0, \dots, m_n \\ m_0+m_1+\cdots+m_{n} + n = p}} \cF^{m_n} \cG \cF^{m_{n-1}} \cdots \cF^{m_1} \cG \cF^{m_0} A
\]
for $p \geq n$ and $\alpha_{p,n} = 0$ for $p < n$.
It satisfies 
\[
\| \alpha_{p,n} \|_{\op} \leq \norm{A}_{\op}\dfrac{(2tC)^p u^n}{p!} \sum_{\substack{m_0, \dots, m_n \\ m_0+m_1+\cdots+m_{n} + n = p}} 1 = \norm{A}_{\op}\dfrac{(2tC)^p u^n}{p!} \binom{p}{n} =\norm{A}_{\op}\dfrac{(2tCu)^n (2tC)^{p-n}}{n!(p-n)!}
\]
where $C=\max(\norm{V}_{\op},\norm{H_0}_{\op})$. Then, the double sum satisfies
\[
\sum_{n,p} \norm{\alpha_{p,n}}_{\op} \leq \norm{A}_{\op}\sum_{n=0}^\infty \dfrac{(2tCu)^n}{n!}\sum_{p=n}^\infty \dfrac{(2tC)^{p-n}}{(p-n)!} = \norm{A}_{\op}e^{2tCu}e^{2tC}.
\]
It is thus an absolutely convergent double series and the sums can be inverted to obtain
\begin{align*}
    A(t) = \sum_{n=0}^{\infty} u^n \sum_{p \geq n} \frac{t^p}{p!} \sum_{\substack{m_0, \dots, m_n \\ m_0+m_1+\cdots+m_{n} + n = p}} \cF^{m_n} \cG \cF^{m_{n-1}} \cdots \cF^{m_1} \cG \cF^{m_0} A,
\end{align*}
which leads to the desired expression.
\end{proof}

\subsection{Series truncation}
\label{sec:approximatedegree}

Given an integer $\kk$, define $A^{\kk}(t)$ as the truncation of Eq.~\eqref{eq:alternativedyson} up to the $\kk$'th power of $u$
\begin{equation}
\label{eq:Aell}
A^\kk(t) = \sum_{n=0}^\kk u^n\sum_{m_0,\dots,m_n=0}^\infty \dfrac{t^{n+\sum_i m_i}}{(n+\sum_i m_i)!}\cF^{m_n}\cG\cF^{m_{n-1}}\cdots \cF^{m_1} \cG \cF^{m_0} A \,.
\end{equation}
Note that, by Eq.~\eqref{eq:degreecommutatormonomial}, $\deg (A^{\kk}(t)) \leq \deg(A) + 2\kk$.
The main goal of this section is to bound the error $\|A^\kk(t) - A(t)\|_{\F}$.
\begin{lemma}
\label{lem:At-Akt}
Consider an arbitrary operator $A\in\cA$. Then, for any $t < \tmax(u) := \frac{\log(e/u)}{8 e^2 \Delta (\deg A+2)}$ and $\kk\in\NN$, we have that
\be
\label{eq:At-Akt-fro}
\|A^\kk(t) - A(t)\|_{\F} \leq \frac{(t/\tmax(u))^{k+1}}{1-t/\tmax(u)} \|A\|_{\F}.
\ee
\end{lemma}
Before proving this lemma, we first prove a bound on the norm of nested commutators contained in Eq.~\eqref{eq:Aell}. Recall that $\cF(\cdot) := \i [H_0,\cdot]$ and $\cG(\cdot) = \i [V,\cdot]$, where $H_0$ is the quadratic part of $H$ and $V$ is the quartic one.
\begin{proposition}[Frobenius norm of nested commutators]
For any tuple $(m_0,\dots,m_n)\in\NN^{n+1}$ and $m = m_0 + \dots + m_n$, and $A \in \cA$,
    \be
    \label{eq:iter-commutator-frob}
        \norm{\cF^{m_n}\cG\cF^{m_{n-1}}\cdots \cF^{m_1} \cG \cF^{m_0} A}_\F \leq \left(4 \Delta (\deg A+2n)\right)^{n+m} \norm{A}_\F.
    \ee
\end{proposition}
\begin{proof}
This is a straightforward application of Theorem \ref{thm:commutatorfrobeniusbound}.
     Let $B= \cF^{m_n}\cG\cF^{m_{n-1}}\cdots \cF^{m_1} \cG \cF^{m_0} A$. Note that by Eq.~\eqref{eq:degreecommutatormonomial}, $\deg(B)\leq d+2n$ where $d=\deg(A)$. Denoting $C(\Delta,d) :=2 \Delta \sqrt{d(d+2)}$ and applying Proposition~\ref{thm:commutatorfrobeniusbound} to the $m_n$ leftmost $\cF$'s, we then have
    \[
        \norm{B}_\F \leq C(\Delta,d+2n)^{m_n}\norm{\cG\cF^{m_{n-1}}\cdots \cF^{m_1} \cG \cF^{m_0} A}_\F\,.
    \]
    We apply the same Proposition to $\cG$, though now the superoperator is acting on a operator of degree $d+2(n-1)$. We then have
    \[
    \begin{aligned}
        \norm{B}_\F &\leq C(\Delta,d+2n)^{m_n}C(\Delta,d+2n-2)\norm{\cF^{m_{n-1}}\cG\cdots \cF^{m_1} \cG \cF^{m_0} A}_\F \\
        &\leq C(\Delta,d+2n)^{m_n+1}\norm{\cF^{m_{n-1}}\cG\cdots \cF^{m_1} \cG \cF^{m_0} A}_\F\,.
    \end{aligned}
    \]
    Repeating the same steps prescribed above, we arrive at 
    \[
        \norm{B}_\F \leq C(\Delta,d+2n)^{n+m} \norm{A}_\F\,.
    \]
    Finally note that $C(\Delta,d+2n)\leq 4\Delta (d+2n)$.
\end{proof}

We are now ready to prove Lemma \ref{lem:At-Akt}.
\begin{proof}[Proof of Lemma \ref{lem:At-Akt}]
We have
    \be
        \norm{A^\kk(t)-A(t)}_{\F} \leq \sum_{n=\kk + 1}^\infty u^n\sum_{m_0,\dots,m_n=0}^\infty \dfrac{t^{n+\sum_i m_i}}{(n+\sum_i m_i)!}\norm{\cF^{m_n}\cG\cF^{m_{n-1}}\cdots \cF^{m_1} \cG \cF^{m_0} A}_{\F}.
    \ee
Using the bound \eqref{eq:iter-commutator-frob},
\be
    \norm{\cF^{m_n}\cG\cF^{m_{n-1}}\cdots \cF^{m_1} \cG \cF^{m_0} A}_{\F} \leq \norm{A}_{\F}(4\Delta (\deg(A) + 2n))^{n+\sum_i m_i},
\ee
and letting $\deg A = d$, we get
    \[
    \begin{aligned}
        \norm{A^\kk(t)-A(t)}_{\F}
        &\leq \norm{A}_{\F}\sum_{n=\kk + 1}^\infty u^n\sum_{m_0,\dots,m_n=0}^\infty \dfrac{(4t\Delta(d+2n))^{n+\sum_i m_i}}{(n+\sum_i m_i)!} \\
        &= \norm{A}_{\F} \sum_{n=\kk + 1}^\infty u^n\sum_{m=0}^\infty \sum_{\substack{m_0,\dots,m_n=0\,: \\\sum_i m_i = m}}^\infty \dfrac{(4t\Delta(d+2n))^{n+m}}{(n+m)!} 
    \end{aligned}
    \]
    where in the last line we partitioned the sum over $(m_0,\dots,m_n)$ into sums over $(m_0,\dots,m_n)$ such that $\sum_i m_i = m$, for $m \in \NN$. We then have
    \[
    \begin{aligned}
         \norm{A^\kk(t)-A(t)}_{\F} &\leq \norm{A}_{\F} \sum_{n=\kk + 1}^\infty u^n\sum_{m=0}^\infty \dfrac{(4t\Delta(d+2n))^{n+m}}{(n+m)!}\sum_{\substack{m_0,\dots,m_n=0\,: \\\sum_i m_i = m}}^\infty 1 \\
        &= \norm{A}_{\F} \sum_{n=\kk + 1}^\infty (4tu\Delta(d+2n))^n\sum_{m=0}^\infty \dfrac{(4t\Delta(d+2n))^m}{(n+m)!} \binom{n+m}{n} \\
        &= \norm{A}_{\F} \sum_{n=\kk + 1}^\infty \dfrac{(4tu\Delta(d+2n))^n}{n!}\sum_{m=0}^\infty \dfrac{(4t\Delta(d+2n))^m}{m!} \\
        &= \norm{A}_{\F} \sum_{n=\kk + 1}^\infty \dfrac{(4tu\Delta(d+2n))^n}{n!}e^{4t\Delta(d+2n)} \\
        &\leq \norm{A}_{\F} \sum_{n=\kk + 1}^\infty \dfrac{n^n\left(4tu\Delta(d+2)e^{4t\Delta (d+2)}\right)^n}{n!} \\
        &\leq \norm{A}_{\F} \sum_{n=\kk + 1}^\infty \left(4 ue 4t\Delta(d+2)e^{4t\Delta (d+2)}\right)^n\\
        &\leq \norm{A}_{\F} \sum_{n=\kk + 1}^\infty \left((u/e) \cdot 4e^2 t\Delta(d+2)e^{4e^2 t\Delta (d+2)}\right)^n \,,
    \end{aligned}
    \]
    where the second inequality is achieved by noting that $(d+2n)\leq n(d+2)$ and the third one via the inequality $n^n/n! \leq e^n$.
    Letting $t^*=\frac{1}{4e^2\Delta (d+2)}$, the series converges when 
    \[
    \frac{u}{e} \cdot
    \dfrac{t}{t^*} e^{t/t^*} < 1.
    \]
    Now we use the following simple proposition
    \begin{proposition}
    If $0 < t < \frac{1}{2} t^* \log(e/u) =: \tmax$ then $\frac{u}{e} \frac{t}{t^*} e^{t/t^*} < \frac{t}{\tmax} < 1$.
    \end{proposition}
        \begin{proof}
    If $t < \tmax$ then $\frac{u}{e} \frac{t}{t^*} e^{t/t^*} < \frac{u}{e} \frac{t}{t^*} \sqrt{\frac{e}{u}} = \frac{t}{t^*\cdot \sqrt{e/u}} \leq \frac{t}{t^* \cdot \log(\sqrt{e/u})} = \frac{t}{\tmax} < 1$.
    \end{proof}
    Hence if $t < \frac{\log(e/u)}{8e^2 \Delta (d+2)} =: \tmax(u)$ we have
    \[
    \norm{A^k(t)-A(t)}_{\F} \leq \|A\|_{\F} \sum_{n=k+1}^{\infty} \left(\frac{t}{\tmax(u)}\right)^n = \|A\|_{\F} \dfrac{\left(\frac{t}{\tmax(u)}\right)^{\kk+1}}{1-\frac{t}{\tmax(u)}} \,
    \]
    as desired.
\end{proof}

\begin{remark}
It is possible to prove a bound of the form \eqref{eq:At-Akt-fro} for the operator norm instead of the Frobenius norm, using essentially the same proof, but with a different argument for the inequality \eqref{eq:iter-commutator-frob}. We omit this here as we focus on the Frobenius norm in this paper.
\end{remark}

\begin{proof}[Proof of Theorem \ref{thm:weak-interaction}]
Using Lemma \ref{lem:At-Akt} and the observation that $\deg A^k(t) \leq \deg A + 2k$, we get with $k=\lfloor (\ell-\deg A)/2 \rfloor$,
\[
\best_{\F}(H,A;t,\ell) \leq \frac{(t/\tmax(u))^{(\ell-\deg A)/2}}{1-t/\tmax(u)} \|A\|_{\F}.
\]
\end{proof}

\subsection{Majorana Propagation for weakly interacting Hamiltonians}
So far we have only shown the \emph{existence} of an operator $P$ satisfying $\deg(P) \leq \deg A + 2\kk$ which converges to $A(t)$ exponentially fast in $\kk$. Note that the operator $P=A^k(t)$ used in the proof (see Equation \eqref{eq:Aell}) involves an infinite series and it is not clear if it can be computed efficiently.

In the following, we prove Corollary \ref{corollary:MPweaklyinteracting}, which states that the MP algorithm is able to find such a low-degree approximation of $A(t)$ with complexity $N^{O(\log(t/\eps))}$. To show this, we use the results of Theorem \ref{thm:weak-interaction} in Theorem \ref{thm:majorana-prop} and then specify the scalings of $\dt, \ell$ with respect to $\eps,t$ for the error of MP to be less than $\eps$.
\begin{proof}[Proof of Corollary \ref{corollary:MPweaklyinteracting}]
We assume that $\|A\|_{\F} = 1$.
From Theorem \ref{thm:intro-weak-interaction} we know that for $t \leq \tmax(u)/2$,
\[
\best_{\F}(H(u),A,t,\ell) \leq 2 \cdot 2^{-(\ell-\deg A)/2}.
\]
Plugging this in the error bound~\eqref{eq:bounderrppell} of the output of MP, we get that the MP algorithm with time step $\dt$ and truncation parameter $\ell$ has error at most
\[
\begin{aligned}
\|A_{\MP}^{\ell,\dt}(t) - A(t)\|_{\F} &\leq 34 t\cdot \dt \cdot \Delta^2 (\ell+2)^2 + 2 \lceil t/\dt \rceil 2^{-(\ell-\deg A)/2}\\
&\leq 34 t\cdot \dt \cdot \Delta^2 (\ell+2)^2 + 2 (t/\dt +1) 2^{-(\ell-\deg A)/2}.
\end{aligned}
\]
We need to choose parameters $\dt$ and $\ell$ such that the right-hand side is at most $\eps$. For a given $\ell$, the optimal $\dt$, which minimizes the right-hand side is
\be
\label{eq:optimal-deltat-weakhamiltonian}
\dt = \sqrt{\frac{2}{34}} \cdot \frac{\sqrt{2^{-(\ell-\deg A)/2}}}{\Delta (\ell+2)}
\ee
for which the error bound becomes
\[
\begin{aligned}
\|A_{\MP}^{\ell,\dt}(t) - A(t)\|_{\F} \leq C_1 t \Delta (\ell+2) 2^{-(\ell-\deg A)/4} + 2 \cdot 2^{-(\ell-\deg A)/2}
\end{aligned}
\]
where $C_1 = 2 \sqrt{34 \cdot 2}$. We find $\ell$ such that both terms are less than $\eps/2$, i.e., we want, with $L = \ell/4$
\begin{align}
L 2^{-L} &\leq \frac{\eps}{16 C_1 t \Delta} 2^{-\deg A / 4} \label{eq:ineq1Lambert}\\
2^{-(4L-\deg A)/2} &\leq \frac{\eps}{4}. \label{eq:ineq2log}
\end{align}
We use the next simple proposition to solve the first inequality.
\begin{proposition}
If $\rho > 0$ and $a > 0$, then for any $v \geq \max(0,a+\log_2(4a/\rho + 4/\rho^2))$,  we have $v2^{-v} \leq \rho 2^{-a}$.
\end{proposition}
\begin{proof}
The function $v\mapsto v2^{-v}$ is increasing on $[0,1/\log 2]$ and decreasing on $[1/\log 2,\infty)$ and takes the (maximal) value $1/(e\log 2) \approx 0.53$ at $v=1/\log 2$. For $y > 0$, define
\[
v_0(y) := \begin{cases}
    0 \quad \text{ if } y \geq 1/(e \log 2)\\
    \text{the unique solution in $[1/\log 2, +\infty)$ to $v_0 2^{-v_0} = y$ if $y < 1/(e\log 2)$}.
\end{cases}
\]
Note that in the second case, this unique solution can be expressed in terms of the Lambert function as $v_0(y) = W_{-1}(-y\log 2)/(-\log 2)$.
 By definition, for any $y > 0$, we have $v 2^{-v} \leq y$ for all $v \geq v_0(y)$. It is not hard to check that for all $y > 0$
\[
v_0(y) \leq 
\begin{cases}
    \max(0,2+\log_2(1/y) + \log_2 \log_2 (1/y)) & \text{ if } y \in [0,1)\\
    0 & \text{ else.}
\end{cases}
\]

With $y = \rho 2^{-a}$, and assuming that $y \in [0,1)$ we have
\[
\begin{aligned}
2 + \log_2 (1/y) + \log_2 \log_2 (1/y) &= 2 + a + \log_2(a \rho^{-1} + \rho^{-1} \log_2(\rho^{-1}))\\
&\leq 2 + a + \log_2(a \rho^{-1} + \rho^{-2})\\
&= a + \log_2(4a\rho^{-1} + 4\rho^{-2})
\end{aligned}
\]
Hence with this we have that for any $\rho, a > 0$
\[
v_0(\rho 2^{-a}) \leq \max(0,a + \log_2(4a \rho^{-1} + 4\rho^{-2})).
\]
This concludes the proof.
\end{proof}
Hence if we take $L \geq \max\left(0,\deg A / 4 + \log_2\left(4\frac{16 C_1 t \Delta \deg A}{\eps} + 4\left(\frac{16 C_1 t \Delta}{\eps}\right)^2 \right)\right)$ the first inequality \eqref{eq:ineq1Lambert} is satisfied. To satisfy the second inequality \eqref{eq:ineq2log}, we can take $L \geq \deg A/4 + \frac{1}{2} \log_2(\frac{4}{\eps})$. Finally, since $\ell = 4L$, we have the desired error bound with
\[
\ell = \deg A + 4 \left \lceil \log_2\left( 4 \frac{\max\left(1 , 16 C_1 t \Delta \deg A + (16 C_1 t \Delta)^2/\eps \right)}{\eps}\right)\right\rceil.
\]
(We used the assumption that $\eps < 4$ which implies that the argument of the logarithm is $\geq 1$, and thus that the whole quantity is necessarily $\geq 0$.) 
With such a choice of $\ell$, note that, from \eqref{eq:optimal-deltat-weakhamiltonian}, $\dt = \eps^{O(1)}$, and so the total complexity of the algorithm is $t/\dt \cdot N^{O(\ell)} = N^{O(\log((1+t)/\eps))}$.
\end{proof}

\section{Numerical experiments}
\label{sec:numerical}
Here we report numerical experiments that demonstrate the applicability of MP in the context of simulation of real-time quantum dynamics. As a well-known benchmark for simulation of fermionic systems, we focus on the Fermi-Hubbard model, described by the following Hamiltonian
\be
    H = -\sum_{\langle i,j\rangle,\sigma} c^*_{i,\sigma} c_{j,\sigma} + c^*_{j,\sigma} c_{i,\sigma} + U\sum_{i} n_{i,\uparrow} n_{i,\downarrow} \,,
\ee
where $c^*_{i,\sigma},c_{i,\sigma}$ are the fermionic creation and annihilation operators. Here $\langle i,j \rangle$ runs over all neighbouring sites and $\sigma \in \{\uparrow,\downarrow\}$ specifies the spin degree-of-freedom. The first sum in the Hamiltonian represents hopping between sites, while the second one represents the Coulomb repulsion between two electrons on the same site. Even though this is a simplified model of electrons in materials, it already contains complex underlying physics and has also been a central object of study for advancing many classical algorithms for many-body systems \cite{leblanc2015solutions}, and more recently quantum algorithms~\cite{arute2020observation,alam2025fermionic,alam2025programmable}.

First, we show the accuracy that MP is able to achieve. In Fig.~\ref{fig:fh005} we plot the distance in Frobenius norm between $A^\ell_\MP(t)$ and $A_{\trotter}(t)$, the latter being the observable obtained by applying trotterized dynamics only (i.e., without truncation). We do this for a one-dimensional system with $N=6$ sites and an interaction strength $U=1$. As expected from Corollary~\ref{corollary:MPweaklyinteracting}, we see exponential convergence of the output of MP towards the ideal observable.
\definecolor{color0}{HTML}{96c24f}
\definecolor{color1}{HTML}{769d5d}
\definecolor{color2}{HTML}{55786b}
\definecolor{color3}{HTML}{355379}
\definecolor{color4}{HTML}{152e87}
\begin{figure}[ht]
\centering
\begin{tikzpicture}
\begin{axis}[
    xlabel={maximal degree $\ell$},
    ylabel={$\norm{A^\ell_\MP(t)-A_{\trotter}(t)}_\F$},
    legend pos=south west,
    ymode=log,
    width=9cm,
    height=7cm
]
\addplot[color=color0,line width = 1.5pt] table [x=deg, y=t0.2, col sep=comma] {data/fh005.csv};
\addlegendentry{$t=0.2$}

\addplot[color=color1,line width = 1.5pt] table [x=deg, y=t0.4, col sep=comma] {data/fh005.csv};
\addlegendentry{$t=0.4$}

\addplot[color=color2,line width = 1.5pt] table [x=deg, y=t0.6, col sep=comma] {data/fh005.csv};
\addlegendentry{$t=0.6$}

\addplot[color=color3,line width = 1.5pt] table [x=deg, y=t0.8, col sep=comma] {data/fh005.csv};
\addlegendentry{$t=0.8$}

\addplot[color=color4,line width = 1.5pt] table [x=deg, y=t1.0, col sep=comma] {data/fh005.csv};
\addlegendentry{$t=1.0$}
\end{axis}
\end{tikzpicture}
\caption{Average error for 1D Fermi-Hubbard model with $6$ sites. We plot the distance in Frobenius norm between the MP algorithm output and a pure Trotter evolution, as a function of the maximal degree $\ell$. Different lines indicate different time horizons. In all cases, there is an exponential suppression of the Frobenius norm as the maximal degree is increased. The time-step chosen for discretizing the dynamics is $\dt = 0.01$.}
\label{fig:fh005}
\end{figure}

In addition, similar to Ref.~\cite{danna2025majorana}, we consider also two-dimensional systems of size $L\times L$ for $L = 3,5,7$\footnote{For additional numerical stability, in these results we also discard strings with negligible coefficients, i.e. strings $\gamma_S$ whose Majorana coefficient satisfies $\abs{a_S} \leq$ 1e-5.}.
We consider an initial state where all the sites are filled in an antiferromagnetic pattern except for the central one which is empty. In other words, letting $N_\text{centre}=(L^2+1)/2$, we consider
\be
    \ket{\psi} = \prod_{i\neq N_\text{centre}} c_{i,\sigma(i)}^* \ket{\text{0}}
    \qquad
    \sigma(i) =
    \begin{cases}
        \uparrow &\text{ if ($i<N_\text{centre}$ and $i$ odd) or ($i>N_\text{centre}$ and $i$ even)} \\
        \downarrow &\text{ if ($i<N_\text{centre}$ and $i$ even) or (if $i>N_\text{centre}$ and $i$ odd)}.
    \end{cases}
\ee
For each system size, we plot in Fig.~\ref{fig:fh-twodimensional} the probability of the central site being empty as a function of time, i.e. we plot the time evolution of $\langle h_{N_\text{centre},\uparrow}h_{N_\text{centre},\downarrow}\rangle$ where $h_{j,\sigma}= 1-c^*_{j,\sigma}c_{j,\sigma}$. We do this for different values of interaction strength $U \in [0,4]$. Additionally, for $L=3$ we also plot the exact solution as this is readily available at this system size. Across all system sizes, at fixed $U$ we see that an increase in the maximal degree $\ell$ generally leads to higher accuracy. At the same time, as we increase $U$ it can be noticed that the plot for a given maximal degree tends to diverge at increasingly earlier times. This is testament of the behaviour of the characteristic maximal time $\tmax(u)$ with respect to the interaction strength $U$.
It is interesting to note that for all values of $L$ and for $U\leq 2$, the $\ell=8,10$ curves are roughly overlapping, which may be caused by the fact that the truncation parameter is already well converged across all these cases (with the deviation from the exact case for $L=3$ being caused by the Trotter approximation).

As a final remark, let us note that during simulation, each newly generated Majorana string $\gamma_S$ that satisfies $|S|>\ell$ is removed after every single Majorana rotation, as opposed to after having gone through the full Trotter sweep. This can only worsen the accuracy as it might be possible that, within a Trotter step, strings surpassing the threshold might recombine to produce lower-degree strings. However, as seen from the numerical results, the overall results remain highly accurate and reflect the analytical results concerning the MP algorithm.

\definecolor{ell4color}{HTML}{BD2327}
\definecolor{ell6color}{HTML}{C3A10C}
\definecolor{ell8color}{HTML}{769d5d}
\definecolor{ell10color}{HTML}{355379}
\begin{figure}[ht]
\centering
\begin{tikzpicture}
\begin{groupplot}[
    group style={
        group size=3 by 2,
        horizontal sep=1.25cm,
        vertical sep=1.25cm,
        ylabels at=edge left,
    },
    width=5cm,
    height=4cm,
    ylabel={$\langle h_{5,\uparrow}h_{5,\downarrow}\rangle$},
    ytick={0.0,0.2,0.4,0.6,0.8,1.0},
    ymin=-0.05,ymax=1.05,
]

\nextgroupplot[title={$3 \times 3$, $U=0.0$},
    legend to name=sharedlegendfh004,
    legend columns=-1,
]
\addplot[color=grey,line width = 1.5pt,style=dashed] table [x=time, y=U0.0exact, col sep=comma] {data/fh004.csv};
\addlegendentry{exact}
\addplot[color=ell4color,line width = 1.5pt] table [x=time, y=U0.0ell4, col sep=comma] {data/fh004.csv};
\addlegendentry{$\ell=4$}
\addplot[color=ell6color,line width = 1.5pt] table [x=time, y=U0.0ell6, col sep=comma] {data/fh004.csv};
\addlegendentry{$\ell=6$}
\addplot[color=ell8color,line width = 1.5pt] table [x=time, y=U0.0ell8, col sep=comma] {data/fh004.csv};
\addlegendentry{$\ell=8$}
\addplot[color=ell10color,line width = 1.5pt] table [x=time, y=U0.0ell10, col sep=comma] {data/fh004.csv};
\addlegendentry{$\ell=10$}

\nextgroupplot[title={$3 \times 3$, $U=0.5$}]
\addplot[color=grey,line width = 1.5pt,style=dashed] table [x=time, y=U0.5exact, col sep=comma] {data/fh004.csv};
\addplot[color=ell4color,line width = 1.5pt] table [x=time, y=U0.5ell4, col sep=comma] {data/fh004.csv};
\addplot[color=ell6color,line width = 1.5pt] table [x=time, y=U0.5ell6, col sep=comma] {data/fh004.csv};
\addplot[color=ell8color,line width = 1.5pt] table [x=time, y=U0.5ell8, col sep=comma] {data/fh004.csv};
\addplot[color=ell10color,line width = 1.5pt] table [x=time, y=U0.5ell10, col sep=comma] {data/fh004.csv};

\nextgroupplot[title={$3 \times 3$, $U=1.0$}]
\addplot[color=grey,line width = 1.5pt,style=dashed] table [x=time, y=U1.0exact, col sep=comma] {data/fh004.csv};
\addplot[color=ell4color,line width = 1.5pt] table [x=time, y=U1.0ell4, col sep=comma] {data/fh004.csv};
\addplot[color=ell6color,line width = 1.5pt] table [x=time, y=U1.0ell6, col sep=comma] {data/fh004.csv};
\addplot[color=ell8color,line width = 1.5pt] table [x=time, y=U1.0ell8, col sep=comma] {data/fh004.csv};
\addplot[color=ell10color,line width = 1.5pt] table [x=time, y=U1.0ell10, col sep=comma] {data/fh004.csv};

\nextgroupplot[title={$3 \times 3$, $U=1.5$}]
\addplot[color=grey,line width = 1.5pt,style=dashed] table [x=time, y=U1.5exact, col sep=comma] {data/fh004.csv};
\addplot[color=ell4color,line width = 1.5pt] table [x=time, y=U1.5ell4, col sep=comma] {data/fh004.csv};
\addplot[color=ell6color,line width = 1.5pt] table [x=time, y=U1.5ell6, col sep=comma] {data/fh004.csv};
\addplot[color=ell8color,line width = 1.5pt] table [x=time, y=U1.5ell8, col sep=comma] {data/fh004.csv};
\addplot[color=ell10color,line width = 1.5pt] table [x=time, y=U1.5ell10, col sep=comma] {data/fh004.csv};

\nextgroupplot[title={$3 \times 3$, $U=2.0$}]
\addplot[color=grey,line width = 1.5pt,style=dashed] table [x=time, y=U2.0exact, col sep=comma] {data/fh004.csv};
\addplot[color=ell4color,line width = 1.5pt] table [x=time, y=U2.0ell4, col sep=comma] {data/fh004.csv};
\addplot[color=ell6color,line width = 1.5pt] table [x=time, y=U2.0ell6, col sep=comma] {data/fh004.csv};
\addplot[color=ell8color,line width = 1.5pt] table [x=time, y=U2.0ell8, col sep=comma] {data/fh004.csv};
\addplot[color=ell10color,line width = 1.5pt] table [x=time, y=U2.0ell10, col sep=comma] {data/fh004.csv};

\nextgroupplot[title={$3 \times 3$, $U=4.0$}]
\addplot[color=grey,line width = 1.5pt,style=dashed] table [x=time, y=U4.0exact, col sep=comma] {data/fh004.csv};
\addplot[color=ell4color,line width = 1.5pt] table [x=time, y=U4.0ell4, col sep=comma] {data/fh004.csv};
\addplot[color=ell6color,line width = 1.5pt] table [x=time, y=U4.0ell6, col sep=comma] {data/fh004.csv};
\addplot[color=ell8color,line width = 1.5pt] table [x=time, y=U4.0ell8, col sep=comma] {data/fh004.csv};
\addplot[color=ell10color,line width = 1.5pt] table [x=time, y=U4.0ell10, col sep=comma] {data/fh004.csv};

\end{groupplot}
\node at (current bounding box.north) [above=1.5mm] {\pgfplotslegendfromname{sharedlegendfh004}};
\end{tikzpicture}

\begin{tikzpicture}
\begin{groupplot}[
    group style={
        group size=3 by 2,
        horizontal sep=1.25cm,
        vertical sep=1.25cm,
        ylabels at=edge left,
    },
    width=5cm,
    height=4cm,
    ylabel={$\langle h_{13,\uparrow}h_{13,\downarrow}\rangle$},
    ytick={0.0,0.2,0.4,0.6,0.8,1.0},
    ymin=-0.05,ymax=1.05,
]

\nextgroupplot[title={$5 \times 5$, $U=0.0$},
    legend to name=sharedlegendfh006,
    legend columns=-1,
]

\addplot[color=ell4color,line width = 1.5pt] table [x=time, y=U0.0ell4, col sep=comma] {data/fh006.csv};
\addlegendentry{$\ell=4$}
\addplot[color=ell6color,line width = 1.5pt] table [x=time, y=U0.0ell6, col sep=comma] {data/fh006.csv};
\addlegendentry{$\ell=6$}
\addplot[color=ell8color,line width = 1.5pt] table [x=time, y=U0.0ell8, col sep=comma] {data/fh006.csv};
\addlegendentry{$\ell=8$}
\addplot[color=ell10color,line width = 1.5pt] table [x=time, y=U0.0ell10, col sep=comma] {data/fh006.csv};
\addlegendentry{$\ell=10$}

\nextgroupplot[title={$5 \times 5$, $U=0.5$}]

\addplot[color=ell4color,line width = 1.5pt] table [x=time, y=U0.5ell4, col sep=comma] {data/fh006.csv};
\addplot[color=ell6color,line width = 1.5pt] table [x=time, y=U0.5ell6, col sep=comma] {data/fh006.csv};
\addplot[color=ell8color,line width = 1.5pt] table [x=time, y=U0.5ell8, col sep=comma] {data/fh006.csv};
\addplot[color=ell10color,line width = 1.5pt] table [x=time, y=U0.5ell10, col sep=comma] {data/fh006.csv};

\nextgroupplot[title={$5 \times 5$, $U=1.0$}]
\addplot[color=ell4color,line width = 1.5pt] table [x=time, y=U1.0ell4, col sep=comma] {data/fh006.csv};
\addplot[color=ell6color,line width = 1.5pt] table [x=time, y=U1.0ell6, col sep=comma] {data/fh006.csv};
\addplot[color=ell8color,line width = 1.5pt] table [x=time, y=U1.0ell8, col sep=comma] {data/fh006.csv};
\addplot[color=ell10color,line width = 1.5pt] table [x=time, y=U1.0ell10, col sep=comma] {data/fh006.csv};

\nextgroupplot[title={$5 \times 5$, $U=1.5$}]
\addplot[color=ell4color,line width = 1.5pt] table [x=time, y=U1.5ell4, col sep=comma] {data/fh006.csv};
\addplot[color=ell6color,line width = 1.5pt] table [x=time, y=U1.5ell6, col sep=comma] {data/fh006.csv};
\addplot[color=ell8color,line width = 1.5pt] table [x=time, y=U1.5ell8, col sep=comma] {data/fh006.csv};
\addplot[color=ell10color,line width = 1.5pt] table [x=time, y=U1.5ell10, col sep=comma] {data/fh006.csv};

\nextgroupplot[title={$5 \times 5$, $U=2.0$}]
\addplot[color=ell4color,line width = 1.5pt] table [x=time, y=U2.0ell4, col sep=comma] {data/fh006.csv};
\addplot[color=ell6color,line width = 1.5pt] table [x=time, y=U2.0ell6, col sep=comma] {data/fh006.csv};
\addplot[color=ell8color,line width = 1.5pt] table [x=time, y=U2.0ell8, col sep=comma] {data/fh006.csv};
\addplot[color=ell10color,line width = 1.5pt] table [x=time, y=U2.0ell10, col sep=comma] {data/fh006.csv};

\nextgroupplot[title={$5 \times 5$, $U=4.0$}]
\addplot[color=ell4color,line width = 1.5pt] table [x=time, y=U4.0ell4, col sep=comma] {data/fh006.csv};
\addplot[color=ell6color,line width = 1.5pt] table [x=time, y=U4.0ell6, col sep=comma] {data/fh006.csv};
\addplot[color=ell8color,line width = 1.5pt] table [x=time, y=U4.0ell8, col sep=comma] {data/fh006.csv};
\addplot[color=ell10color,line width = 1.5pt] table [x=time, y=U4.0ell10, col sep=comma] {data/fh006.csv};

\end{groupplot}
\end{tikzpicture}
 
\begin{tikzpicture}
\begin{groupplot}[
    group style={
        group size=3 by 2,
        horizontal sep=1.25cm,
        vertical sep=1.2cm,
        ylabels at=edge left,
        xlabels at=edge bottom,
    },
    width=5cm,
    height=4cm,
    xlabel={time horizon $t$},
    ylabel={$\langle h_{25,\uparrow}h_{25,\downarrow}\rangle$},
    ytick={0.0,0.2,0.4,0.6,0.8,1.0},
    ymin=-0.05,ymax=1.05,
]

\nextgroupplot[title={$7 \times 7$, $U=0.0$},
    legend to name=sharedlegendfh007,
    legend columns=-1,
]

\addplot[color=ell4color,line width = 1.5pt] table [x=time, y=U0.0ell4, col sep=comma] {data/fh007.csv};
\addlegendentry{$\ell=4$}
\addplot[color=ell6color,line width = 1.5pt] table [x=time, y=U0.0ell6, col sep=comma] {data/fh007.csv};
\addlegendentry{$\ell=6$}
\addplot[color=ell8color,line width = 1.5pt] table [x=time, y=U0.0ell8, col sep=comma] {data/fh007.csv};
\addlegendentry{$\ell=8$}
\addplot[color=ell10color,line width = 1.5pt] table [x=time, y=U0.0ell10, col sep=comma] {data/fh007.csv};
\addlegendentry{$\ell=10$}

\nextgroupplot[title={$7 \times 7$, $U=0.5$}]

\addplot[color=ell4color,line width = 1.5pt] table [x=time, y=U0.5ell4, col sep=comma] {data/fh007.csv};
\addplot[color=ell6color,line width = 1.5pt] table [x=time, y=U0.5ell6, col sep=comma] {data/fh007.csv};
\addplot[color=ell8color,line width = 1.5pt] table [x=time, y=U0.5ell8, col sep=comma] {data/fh007.csv};
\addplot[color=ell10color,line width = 1.5pt] table [x=time, y=U0.5ell10, col sep=comma] {data/fh007.csv};

\nextgroupplot[title={$7 \times 7$, $U=1.0$}]
\addplot[color=ell4color,line width = 1.5pt] table [x=time, y=U1.0ell4, col sep=comma] {data/fh007.csv};
\addplot[color=ell6color,line width = 1.5pt] table [x=time, y=U1.0ell6, col sep=comma] {data/fh007.csv};
\addplot[color=ell8color,line width = 1.5pt] table [x=time, y=U1.0ell8, col sep=comma] {data/fh007.csv};
\addplot[color=ell10color,line width = 1.5pt] table [x=time, y=U1.0ell10, col sep=comma] {data/fh007.csv};

\nextgroupplot[title={$7 \times 7$, $U=1.5$}]
\addplot[color=ell4color,line width = 1.5pt] table [x=time, y=U1.5ell4, col sep=comma] {data/fh007.csv};
\addplot[color=ell6color,line width = 1.5pt] table [x=time, y=U1.5ell6, col sep=comma] {data/fh007.csv};
\addplot[color=ell8color,line width = 1.5pt] table [x=time, y=U1.5ell8, col sep=comma] {data/fh007.csv};
\addplot[color=ell10color,line width = 1.5pt] table [x=time, y=U1.5ell10, col sep=comma] {data/fh007.csv};

\nextgroupplot[title={$7 \times 7$, $U=2.0$}]
\addplot[color=ell4color,line width = 1.5pt] table [x=time, y=U2.0ell4, col sep=comma] {data/fh007.csv};
\addplot[color=ell6color,line width = 1.5pt] table [x=time, y=U2.0ell6, col sep=comma] {data/fh007.csv};
\addplot[color=ell8color,line width = 1.5pt] table [x=time, y=U2.0ell8, col sep=comma] {data/fh007.csv};
\addplot[color=ell10color,line width = 1.5pt] table [x=time, y=U2.0ell10, col sep=comma] {data/fh007.csv};

\nextgroupplot[title={$7 \times 7$, $U=4.0$}]
\addplot[color=ell4color,line width = 1.5pt] table [x=time, y=U4.0ell4, col sep=comma] {data/fh007.csv};
\addplot[color=ell6color,line width = 1.5pt] table [x=time, y=U4.0ell6, col sep=comma] {data/fh007.csv};
\addplot[color=ell8color,line width = 1.5pt] table [x=time, y=U4.0ell8, col sep=comma] {data/fh007.csv};
\addplot[color=ell10color,line width = 1.5pt] table [x=time, y=U4.0ell10, col sep=comma] {data/fh007.csv};

\end{groupplot}
\end{tikzpicture}

\caption{Expectation value of the hole density in central site for $L\times L$ Fermi Hubbard model with $L\in\{3,5,7\}$. We plot, for different values of the interaction strength $U$, the behaviour of the hole density in the central site as a function of time. We compare across multiple values of $\ell\in\{4,6,8,10\}$ to show in a qualitative way that higher $\ell$ converges to the exact result. The time-step is $\dt =0.02$.}
\label{fig:fh-twodimensional}
\end{figure}

\clearpage
\newpage

\section*{Acknowledgments}
G.F. and H.F. acknowledge funding by UK Research and Innovation (UKRI) under the UK government’s Horizon Europe funding guarantee EP/X032051/1. O.F. acknowledges support by the European Research Council (ERC Grant AlgoQIP, Agreement No. 851716), by the European Union’s Horizon 2020 research and innovation programme under Grant Agreement No 101017733 (VERIqTAS) and by the Agence Nationale de la Recherche under the Plan France 2030 with the reference ANR-22-PETQ-0007.
\section*{Data availability}
The numerical experiments of this work were performed using the Python package FastFermion available at the GitHub repository \url{https://github.com/hfawzi/fastfermion}.
\bibliography{./fermibiblio}

\clearpage
\newpage

\appendix
\section{Bound via Lieb-Robinson}
\label{appendix:lieb-robinson}

In this appendix we derive general upper bounds on
\[
    \best(H,A,t,\ell) := \max_{s\in [0,t]} \min_{\substack{P\in\cA \\
    \deg(P)\leq \ell}} \norm{A(t)-P}
\]
using Lieb-Robinson theorems. In order to state the main result of this appendix, we first define the interaction graph of the Hamiltonian $H = \sum_{X \in \cX} \hc_X \gamma_X$ as the graph on vertices $V=\{1,\ldots,N\}$ (Majorana indices) and where vertices $x,y \in \{1,\ldots,N\}$ are connected by an edge if there is a Hamiltonian term $X \in \cX$ such that $x,y \in X$. We call $d_H:V\times V \to \RR$ the distance induced on this graph. Given $x \in V$ and $r > 0$ we denote $B(x,r) := \{y \in V : d_H(x,y) \leq r\}$ and let $\omega(r)$ be the largest volume of a ball of radius $r$:
\be
\omega(r) := \max_{x\in V}\abs{B(x,r)}.
\ee
We also define the inverse function of $\omega$ as
\[
\omega^{-1}(v) = \max\{ r \in \NN : \omega(r) \leq v \}.
\]
Note that if the interaction graph is a $D$-dimensional grid, then $\omega(r) = O(r^D)$. As is the case throughout the paper, we assume the Hamiltonian is $\Delta$-sparse, see \eqref{eq:vertexbounded}. The main proposition we prove in this appendix is the following.

\begin{proposition}
\label{prop:lr-approximation-polynomial}
Consider a Hamiltonian $H$ as defined above and assume $\deg H = \r$. Let $A=\sum_{X\subset [N]} a_X\gamma_X$ be an operator expressed as a Majorana polynomial, with $\norm{\va}:=\sum_{X\subset V} \abs{a_X}$. For any $\ell \geq \deg A$, the following upper bound holds
\be
\label{eq:bound-best-lr}
    \best_\op(H,A,t,\ell) \leq c \norm{\va} \deg A \exp\left(v_\textup{LR} t-\omega^{-1}\left(\dfrac{\ell}{\deg A}\right)\right)
\ee
where $v_\textup{LR} \leq 2eh\Delta \r^2$ is the Lieb-Robinson velocity and $c=2/\r$.
\end{proposition}

To prove this proposition, we use the following version of the Lieb-Robinson bound.

\begin{proposition}[{Lieb-Robinson bound\footnote{While the result stated in \cite{bentsen2019fast} applies to Pauli Hamiltonians, it can be extended in a straightforward way to Majorana Hamiltonians.}, \cite{bentsen2019fast}}]
\label{prop:lieb-robinson}
    Let $H$ be a $\Delta$-sparse Hamiltonian with $\deg H = \r$ and let $X,Y\subset \{1,\ldots,N\}$ arbitrary. For any time $t \geq 0$, define
    \be
    \label{eq:CXYt}
        C_{XY}(t) := \sup_{\substack{A,\, \supp(A) \subseteq X \\ B,\,\supp(B) \subseteq Y}}\dfrac{\norm{[A(t),B]}_\op}{\norm{A}_\op \norm{B}_\op}.
    \ee
    Then, we have that
    \[
        C_{XY}(t) \leq c\min\{\abs{X},\abs{Y}\} e^{-d_H(X,Y)}\left(e^{v_\textup{LR}t}-1\right)
    \]
    where $c=2/\r$ and $v_\textup{LR} \leq 2e h\Delta \r^2$.
\end{proposition}

\begin{proof}[Proof of Proposition \ref{prop:lr-approximation-polynomial}]

We start by assuming that $A$ consists of a single Majorana term, i.e., $A = \gamma_X$ for some $X \subset \{1,\ldots,N\}$. Consider the set
\[
    \Lambda := \bigcup_{x\in X} B(x,r) \,,
\]
for some $r > 0$ to be chosen later.
By definition of $\Lambda$ we have $d_H(X,\Lambda^c) \geq r$.
    Consider the operator $\Tr_{\Lambda^c}(A(t))$ which is supported on $\Lambda$, and note that it can be expressed as
    \[
    \Tr_{\Lambda^c}(A(t)) = \int_{U_{\Lambda^c}} dU_{\Lambda^c} U_{\Lambda^c} ^*A(t) U_{\Lambda^c}
    \]
    where the integral is over unitaries supported on $\Lambda^c$. Thus, we have
    \be
    \label{eq:ineq-lr-approx-monomial}
    \begin{aligned}
        \norm{A(t)-\Tr_{\Lambda^c}(A(t))}_\op &= \norm{\int \di U_{\Lambda^c} (A(t) - U_{\Lambda^c} ^*A(t) U_{\Lambda^c})}_\op \\
        &\leq \int \di U_{\Lambda^c} \norm{A(t) - U_{\Lambda^c}^*A(t)U_{\Lambda^c}}_\op \\
        &= \int \di U_{\Lambda^c} \norm{U_{\Lambda^c}^*[A(t),U_{\Lambda^c}]}_\op \\
        &= \int \di U_{\Lambda^c} \norm{[A(t),U_{\Lambda^c}]}_\op \\
        &\leq \sup_{\substack{B :\,\supp(B)\subseteq \Lambda^c \\ \norm{B}_\op=1}} \norm{[A(t),B]}_\op \\
        &\leq \norm{A}_\op C_{X\Lambda^c}(t)\\
        &\leq c \deg(A) e^{v_\textup{LR}t-r}
    \end{aligned}
    \ee
    where we used Proposition \ref{prop:lieb-robinson} in the last inequality.
    This shows that $A(t)$ can be approximated by a polynomial of degree at most $|\Lambda| \leq \deg (A) \omega(r)$ with an error bounded by $c \deg(A) e^{v_\text{LR} t - r}$. In other words,
    \[
    \best(t,\deg (A) \omega(r)) \leq \deg(A) e^{v_\text{LR} t - r}.
    \]
    This is true for any $r > 0$. If we choose $r = \omega^{-1}(\ell / \deg A)$ we have $\deg(A)\omega(r) \leq \ell$ and so
    \[
    \best(t,\ell) \leq \best(t,\deg(A) \omega(\omega^{-1}(\ell/\deg A))) \leq \deg(A) e^{v_{\text{LR}}t - \omega^{-1}(\ell/\deg A)}
    \]
    as desired.

It remains to prove the inequality \eqref{eq:bound-best-lr} when $A$ is a general Majorana polynomial. Assume $A = \sum_{X \subset [N]} a_X \gamma_X$, and consider the following approximant of $A(t)$:
    \[
        \tilde{A}(t) := \sum_{X\subset V} a_X \Tr_{\Lambda_X^c}(\gamma_X(t))
    \]
where for each $X$, $\Lambda_X = \cup_{x \in X} B(x,r_X)$ with $r_X = \omega^{-1}(\ell/|X|)$. Using the analysis above for Majorana monomials, we have
    \[
    \begin{aligned}
        \norm{A(t)-\tilde{A}(t)}_\op &\leq \sum_{X\subset V} \abs{a_X} \norm{\gamma_X(t) -\Tr_{\Lambda^c_X}(\gamma_X(t))}_\op \\
        &\leq \sum_{X\subset V} \abs{a_X}c \abs{X}\exp \left(v_\textup{LR}t - \omega^{-1}\left(\dfrac{\ell}{\abs{X}}\right) \right) \\
        &\leq c\norm{\va} \deg A \exp \left(v_\textup{LR}t - \omega^{-1}\left(\dfrac{\ell}{\deg A}\right) \right).
    \end{aligned}
    \]
This completes the proof.
\end{proof}

\section{Analysis for spin systems}
\label{appendix:pauli-propagation}
In what follows, we show that the same results of the main text can be applied in an analogous way to the Pauli case as well, generalizing them to Pauli Propagation (PP).

In this section we define the Pauli algebra and introduce the PP algorithm, in Appendix \ref{appendix:PP-error-analysis} we give an upper bound on the PP error and finally in Appendix \ref{appendix:spins-weak-interaction} we specialize to the weakly interacting regime.
\subsection{Pauli algebra}
Let $\cP_1=\{\Id,X,Y,Z\}$ be the Pauli group associated to a single spin, where
\be
\Id :=\begin{pmatrix}
1 & 0 \\
0 & 1 
\end{pmatrix}\,,\qquad
X :=\begin{pmatrix}
0 & 1 \\
1 & 0 
\end{pmatrix}\,,\qquad
Y :=\begin{pmatrix}
0 & -\ci \\
\ci & 0 
\end{pmatrix}\,,\qquad
Z :=\begin{pmatrix}
1 & 0 \\
0 & -1 
\end{pmatrix}\,,
\ee
these provide a basis for two-dimensional complex matrices. When considering a system of $N$ spins, we then take $\cP_N$ to be the set whose elements are the $N$-fold tensor product of the elements in $\cP_1$. Specifically:
\be
\cP_N := \left\{P \in \left(\CC^{2 \times 2}\right)\opower{N}\,:\,P = S_1 \otimes \dots \otimes S_N\,, S_i \in \cP_1\quad\forall\,i\in[N]\right\}\,,
\ee
where $[N]:=\{1,\dots,N\}$ and each element of this group is often referred to as a \textit{Pauli string}. Note that $\abs{\cP_N} = 4^N$. We will often omit the tensor product notation and simply write $P=S_1\dots S_n$. We will also omit the operators that act trivially, e.g. $X\otimes Y\otimes I\otimes \dots\otimes I \equiv X_1Y_2$. 
Define $\supp(P) \subseteq [N]$ as the sites where $P$ acts non-trivially, and $\abs{P} = \abs{\supp(P)}$. For any $k\in\NN$, we denote the following subsets
\be
\cP_{N,k} := \{P\in\cP_N\,:\,\abs{P}=k\} \,,\qquad \cP_{N,\leq k} := \bigcup_{s=1}^k\cP_{N,s} \,,
\ee
We then consider the algebra $\cA$ of $N$ spins as the span of the $N$-fold Pauli group. Any element $A\in \cA$ can be expressed as a \emph{Pauli polynomial}, i.e. a unique linear combination of Pauli strings
\be\label{eq:pauliobs}
A = \sum_{P\in\cP_N} a_P P \,, \qquad a_P \in \CC \,.
\ee
In this basis, we also have a well-defined notion of \emph{degree} as
\[
    \deg(A) = \max\{|P| : a_P \neq 0\}
\]
 Additionally, an operator $A\in\cA$ is said to be $k$-local if $\deg(A)\leq k$. We furthermore define the operator norm as $\norm{A}_{\op}$ and the (normalized) Frobenius norm $\norm{A}_\F = \tr(A^* A)^{1/2}$ where $\tr$ is the normalized trace operator, which satisfies $\tr(I)=1$.

In the following, we recollect a few properties about Pauli operators
\begin{proposition}[Pauli properties]
\label{prop:pauli-properties}
    Given $P,Q\in\cP_N$, we have:
    \begin{itemize}
        \item Pauli operators are orthogonal, in the sense that $\tr(PQ) = \delta_{PQ}$. This also implies that for any $A\in\cA$, given the expression in the Pauli basis \eqref{eq:pauliobs}, its Frobenius norm is given by
        \[
            \norm{A}_{\F} = \left( \sum_{P\in\cP_N} \abs{a_P}^2 \right) ^{1/2}
        \]
        \item Given $\theta \in\RR$, then
        \[
            e^{i\theta P /2}Qe^{-i\theta P /2} =
            \begin{cases}
            Q  &\textup{ if } [P,Q] = 0\,, \\
            \cos(\theta)Q +\ci \sin(\theta) PQ &\textup{ if } [P,Q] \neq 0\,.
            \end{cases}
        \]
    \end{itemize}
\end{proposition}

\subsection{The Hamiltonian}
For $\r \geq 2$, we take $H$ to be some $\r$-local Hamiltonian. Specifically, given a subset $\cX \subset \cP_{N,\leq \r}$, define
\[
    H =\sum_{P\in\cX} h_P P \,,\quad\abs{h_P} \leq 1\,.
\]
Similarly to the fermionic case, we furthermore suppose the Hamiltonian is $\Delta$-\emph{sparse}, meaning that any site $x\in[N]$ appears in no more than $\Delta$ strings in $H$. In other words,
\[
    \max_{x\in [N]} |\{P\in\cP_N\,:\, h_P \neq 0\text{ and }x\in\supp(P)\}| \leq \Delta
\]
In similar spirit to Proposition ~\ref{prop:bounded-degree-property}, we have the following straightforward property:
\begin{proposition}
\label{prop:sparsity-property-pauli}
    For any Pauli string $P\in\cP_N$, the number of terms $Q\in\cX$ that do not commute with $Q$ is bounded by $\Delta\abs{P}$. In other words,
    \[
        \abs{\{Q\in\cX\,:\,[Q,P] \neq 0\}} \leq \Delta\abs{P}
    \]
\end{proposition}
\begin{proof}
    We proceed in a similar way to Proposition ~\ref{prop:bounded-degree-property}. Denoting $P = P_1\dots P_{\abs{P}}$ where $\supp(P_x)=x$, we have
    \[
        [P,Q] = \sum_{x\in\supp(P)} P_{1}\dots,P_{x-1}[P_x,Q]P_{x+1}\dots P_{\abs{P}} \,,
    \]
    so that
    \[
        \abs{\{Q\in\cX\,:\,[Q,P] \neq 0\}} \leq \sum_{x\in\supp(P)} \abs{\{Q\in\cX\,:\,[Q,P_x] \neq 0\}} \leq \Delta|P|
    \]
    where the last inequality is because, by the $\Delta$-sparsity, there are at most $\Delta$ terms in the Hamiltonian which are supported on any given $x$.
\end{proof}
The Hamiltonian generates a unitary dynamics $\{e^{\ci H t}\}_{t\in\RR}$. Given some observable $A\in\cA$, its time evolution in the Heisenberg picture is given by
\be
\label{eq:A(t)}
    A(t) = e^{\ci H t}Ae^{-\ci H t}= \sum_{n=0}^\infty \dfrac{(\ci t)^n}{n!} [H,\dots,[H,[H,A]]\dots] \,.
\ee
we will sometimes use the notation $\tau_t(A)$ for $A(t)$. We now revisit Proposition \ref{prop:partitionHamiltonian} for the case of a $\r$-local, $\Delta$-sparse Pauli Hamiltonian.
\begin{proposition}
\label{prop:pauli-partitionHamiltonian}
There is an efficient algorithm that can compute a partition of the terms $\cX$ of $H$ in $G \leq \r\Delta$ groups
\begin{equation}
\label{eq:Xpartition-pauli}
\cX = \cX^1 \cup \dots \cup \cX^G
\end{equation}
such that for any $i \in \{1,\ldots,G\}$ and any $X,Y \in \cX^i$, $X\cap Y = \emptyset$.
\end{proposition}
\begin{proof}
Consider the graph $\cG=(\cX,E)$ where each node corresponds to a term of the Hamiltonian, and where an edge exists between two terms $X,Y \in \cX$ iff 
$X\cap Y \neq \emptyset$.
Finding a partition such as the one in Eq.~\eqref{eq:Xpartition-pauli} amounts to finding a coloring of $\cG$. Denoting by $\deg_{\cG}(X)$ the number of edges at node $X\in\cX$, we have that
\[
\deg_\cG(X) \leq \sum_{x\in X} \abs{\{Y\in\cX : x \in Y\text{ and } Y\neq X\}} \leq \sum_{x\in X} (\Delta  - 1) \leq \r(\Delta-1)\,,
\]
since by assumption on the set $\cX$, $\abs{X}\leq \r$. A standard result~\cite{diestel2025graph} is that a simple greedy algorithm can color the graph using $\deg_\cG(X) + 1 \leq \r\Delta$ colors.
\end{proof}

\subsection{Pauli Propagation algorithm}
In this section we summarize the Pauli Propagation (PP) algorithm, which works in an equivalent way to MP.
\paragraph{Trotter approximation}
The PP algorithm starts by decomposing the Hamiltonian into $G$ partitions: 
\[
H= H_1+\dots+H_G\,,\quad H_g = \sum_{P\in\cX_g} h_P P\quad\forall\, g\in [G] \,,
\]
where we have that $\cX=\cX_1\cup\dots\cup\cX_G$ and the elements in any set $\cX_g$ are mutually disjoint, i.e. for any $g\in [G]$ and $P,Q\in\cX_g$ then $\supp(P)\cap\supp(Q) = \emptyset$. Remark that, by Proposition \ref{prop:pauli-partitionHamiltonian}, we can partition $\cX$ in $G\leq \r\Delta$ such partitions. For any $A\in\cA$ and $\dt>0$, we denote
\[
    \tau^g_{\dt}(A) := e^{\ci H_g \dt} A e^{-\ci H_g \dt} \,.
\]
We then let the \emph{Trotter approximation} be the composition of all $\{\tau^g_{\dt}(A)\}_{g\in [G]}$, specifically
\be
\label{eq:trotterstep}
    \trotterstep{\dt}(A) = \left(\tau^G_{\dt} \circ \dots \circ \tau^1_{\dt}\right)(A)
\ee
Note that since any $P,Q$ in any given $\cX_g$ are disjoint, they also satisfy $[P,Q]=0$. This means that
\[
    \tau^g_{\dt}(\cdot) = \prod_{P\in\cX_g} e^{i\dt h_P P} \ \big(\cdot\big) \prod_{P\in\cX_g} e^{-i\dt h_P P}
\]
implying that Eq.~\eqref{eq:trotterstep} corresponds to the first-order Trotter approximation.

\paragraph{Degree truncation} Given the exponential size of $\cP_N$, the second key step of PP is a map that keeps the total number of strings stored low. In particular, define $\Trunc_\ell : \cA \rightarrow \cA$ as
\[
    A=\sum_{P\in\cP_N} a_P P \quad\longmapsto\quad \Trunc_\ell(A) =\sum_{P\in\cP_{N,\leq \ell}} a_P P 
\]
i.e. a map that, in the Pauli basis, removes all the Pauli strings $P$ that satisfy $\deg(P) > \ell$.
This map has the following important property
\[
    \norm{\Trunc_\ell(A)}_{\F} \leq \norm{A}_{\F} \,,
\]
which is a direct consequence of Proposition~\ref{prop:pauli-properties}. The PP algorithm then successively iterates these two operations in order to output a low-degree approximation to the Heisenberg evolution. In particular, given an operator $A\in\cA$, a time horizon $t$, an $\ell\in\NN$ and a small time-step $\dt$, PP then outputs the following operator
\be
\label{eq:A_PP(t)}
    A_\PP(t) = \left(\Trunc_\ell\circ \trotterstep{\dt}\right)^{t/\dt}(A) \,.
\ee
assuming $t/\dt$ is an integer.

\bigskip
We would then like to know if and when is it possible to accurately approximate $A(t)$ with PP. To this end, we first introduce the central quantity of our work. Given $\ell \in \NN$, we define the error of the \emph{best} degree-$\ell$ approximant as
\be
    \best_\F(t,\ell) =\max_{s\in[0,t]} \min_{P\in\cA\,:\,\deg(P)\leq \ell} \norm{P-A(t)}_\F \,.
\ee
Whenever $\best_\F(t,\ell)$ is small, it signals the \emph{existence} of a low-degree approximation to the real-time evolution of $A$. Note that in the Frobenius norm, it is clear that for any time $t$, the best degree-$\ell$ approximant is given by $P^\star=\Trunc_\ell(A(t))$. However, in general we do not know how to compute $P^\star$. One might therefore hope that $A_{\PP}(t)$, given an appropriate choice of parameters, produces an accurate approximation to $P^\star$.

\section{Pauli Propagation analysis}
\label{appendix:PP-error-analysis}
We state here the main result concerning the error of PP
\begin{theorem}[Pauli Propagation error analysis]
\label{thm:PP-error-analysis}
    Consider a $\r$-local, $\Delta$-sparse Pauli Hamiltonian $H$ and $A\in\cA$ an arbitrary observable. For any time horizon $t \geq 0$, let $A_\PP(t)$ be the output of the PP algorithm with parameters $\ell$, $\dt$ (as defined in \eqref{eq:A_PP(t)}). Then, we have the following
    \begin{itemize}
        \item[(i)] the operator $A_\PP(t)$ can be computed with complexity $\frac{t}{\dt}N^{O(\ell)}$;
        \item[(ii)] the error is bounded as
        \[
            \norm{A_\PP(t)-A(t)}_{\F} \leq 2 t\dt \Delta^2(\r^2+1)(\r-1+\ell)^2\norm{A}_\F + \left\lceil \dfrac{t}{\dt}\right\rceil \best_\F(t,\ell)
        \]
    \end{itemize}
\end{theorem}
Before proceeding to the proof, which is deferred to Section \ref{sec:PP-error-proof}, we prove some important technical results which will later be useful.

\subsection{Bound on the norm of commutator with $H$}
We first prove a bound on the norm of the commutator between $H$ and arbitrary observable.
\begin{theorem}[Bound on Hamiltonian commutator norm]
\label{thm:commutator-norm-bound}
    Let $A\in\cA$ be an arbitrary observable and $H$ a $\Delta$-sparse, $\r$-local Hamiltonian. Then, the following bound holds
    \[
        \norm{[H,A]}_{\F} \leq 2\Delta\sqrt{\deg(A)(
    \r-1 + \deg(A))}\norm{A}_\F
    \]
\end{theorem}
\begin{proof}
    Express $A$ in the Pauli basis $A=\sum_{P\in\cP_N} a_P P$. Then, the commutator with $H$ is given by 
    \[
        [H,A] = \sum_{P\in\cX}\sum_{Q\in\cP_N} h_P a_Q [P,Q] = \sum_{P\in\cX}\sum_{\substack{Q\in\cP_N\,: \\
        [P,Q]\neq 0}} 2h_P a_Q PQ \,.
    \]
    Its Frobenius norm is then given by
    \be
    \label{eq:||[H,A]||^2_F}
        \norm{[H,A]}^2_\F = \sum_{S\in\cP_N}\abs{\sum_{(P,Q)\in\cC(S)} 2 h_P a_Q f_{P,Q}}^2\,,
    \ee
     where the set $\cC(S)$ is defined as
    \[
        \cC(S) = \{(P,Q)\in\cX \times \cP_N\,:\,[P,Q]\neq 0\text{ and }PQ=f_{P,Q} S\} \,.
    \]
    The coefficients satisfy $f_{P,Q}\in\CC$ and $\abs{f_{P,Q}}=1$ (i.e. they are simply a phase). Let us compute the size of $\cC(S)$. Noting that $PQ=f_{P,Q} S \iff Q = f_{P,Q} PS$, we have
    \[
    \begin{aligned}
        \abs{\cC(S)} &= \abs{\{P\in\cX\,:\,[P, PS] \neq 0\}} \\
        &= \abs{\{P\in\cX\,:\, [P,S] \neq 0\}} \\
        &\leq \Delta\abs{S}
    \end{aligned}
    \]
    where the last inequality is a consequence of Proposition \ref{prop:sparsity-property-pauli}. Moreover note that for $(P,Q) \in \cC(S)$, $\abs{S}=\abs{PQ}\leq \abs{P}+\abs{Q}-1\leq(\r-1)+\abs{Q}$ where the $-1$ term is because for $[P,Q]$ to be non-vanishing, we need $\abs{\supp(P)\cap\supp(Q)} \geq 1$. Returning to \eqref{eq:||[H,A]||^2_F}, we can apply the Cauchy-Schwarz inequality and the fact that $\abs{h_P}\leq 1$ to get
    \[
    \begin{aligned}
        \norm{[H,A]}^2_\F &\leq 4\sum_{S\in\cP_N} \abs{\cC(S)}\sum_{(P,Q)\in\cC(S)} \abs{a_Q}^2 \\
        &\leq 4\Delta \sum_{Q\in\cP_N}(\r-1+\abs{Q})\abs{a_Q}^2\sum_{\substack{P\in\cX\,: \\ [P,Q]\neq 0}} 1\\
        &\overset{(a)}{\leq} 4\Delta^2\sum_{Q\in\cP_N}(\r-1+\abs{Q})\abs{Q} \abs{a_Q}^2 \\
        &\overset{(b)}{\leq} 4\Delta^2 (\r-1+\deg(A))\deg(A) \norm{A}^2_\F
    \end{aligned}
    \]
    where in $(a)$ we used again Proposition \ref{prop:sparsity-property-pauli} and in $(b)$ we used that $\abs{Q}\leq \deg(A)$ for any $a_Q\neq 0$. Taking the square root on both sides, this concludes the proof.
\end{proof}

\subsection{Trotter error}
Similarly to the fermionic case, we can also obtain a bound on the Trotter error which, for any observable $A$ of bounded degree, is independent of the system size. In particular, we show the following
\begin{theorem}
\label{thm:pauli-trotter-error}
    Consider a $\r$-local, $\Delta$-sparse Hamiltonian and denote the corresponding Heisenberg evolution as $\tau_t(\cdot)$. Let $\trotterstep{t}(\cdot)$ be its Trotter approximation as defined in \eqref{eq:trotterstep}. Then, for any $A\in\cA$ with $\deg(A)=d$, we have that
    \[
        \norm{\tau_t(A)-\trotterstep{t}(A)}_\F \leq 2(\Delta t)^2(\r-1+d)^2 (G^2+1)
    \]
\end{theorem}
\begin{remark}
    If each of the $G$ blocks contains operators supported on mutually disjoint sets, then the Trotter error reduces to
    \[
        \norm{\tau_t(A)-\trotterstep{t}(A)}_\F \leq 2 t^2(\r-1+d)^2 (G^2+\Delta^2)
    \]
\end{remark}
\begin{proof}
    The proof is analogous to Theorem \ref{thm:trotter-error}. By Lemma \ref{lemma:taylor-remainder}, for any set of finite-dimensional super-operators $\{\cL_1,\dots,\cL_G\}$ acting on $\cA$ we have that
    \[
    \begin{aligned}
    e^{t\cL_G} \cdots e^{t\cL_1} &= \id + t (\cL_G+\cdots+\cL_1)\\
    & \quad + \sum_{i=1}^G \int_{0}^{t} \di \tau e^{t\cL_G} \cdots e^{t\cL_{i+1}} e^{\tau \cL_i} \cL_i \left( (t-\tau) \cL_i + t \sum_{j=1}^{i-1} \cL_{j} \right).
    \end{aligned}
    \]
    Denote $\cL=\ci[H,\cdot]$ $\cL_g = \ci [H_g,\cdot]$ for $g\in [G]$. Using Taylor’s theorem with integral remainder, we can expand $e^{t\cL}$ to $1^{\text{st}}$ order 
    \[
        e^{t \cL} = \id + t\cL + \int_{0}^{t} \di \tau (t-\tau) e^{\tau \cL} \cL^2 \,.
    \]
    Putting things together, we then get that the Trotter error is bounded as
    \be
    \begin{aligned}
        \norm{\tau_t(A)-\trotterstep{t}(A)}_\F &\leq \norm{\int_{0}^{t} \di \tau (t-\tau) e^{\tau \cL} \cL^2}_\F \\
        &+\sum_{i=1}^G\int_{0}^{t} \di \tau \norm{e^{t\cL_G} \cdots e^{t\cL_{i+1}} e^{\tau \cL_i} \cL_i \left( (t-\tau) \cL_i + t \sum_{j=1}^{i-1} \cL_{j} \right)}_\F
    \end{aligned}
    \ee
    Note that for any $X\in\cA$ we have $\norm{e^{t\cL_g}(X)}_\F=\norm{X}_\F$ for all $g\in [G]$, as well as for $e^{t\cL}$. We then have
    \be
        \norm{\tau_t(A)-\trotterstep{t}(A)}_{\F} \leq \frac{t^2}{2} \left( \norm{\cL^2(A)}_{\F} + \sum_{i=1}^{G} \norm{\cL_i^2(A)}_{\F} + 2 \sum_{1\leq j < i \leq G} \norm{\cL_i \cL_j(A)}_{\F} \right).
    \ee
    We conclude the proof noting that by Proposition~\ref{thm:commutator-norm-bound}, for all $i,j\in [G]$ we have
    \[
        \begin{aligned}
            \norm{\cL_i\cL_j(A)}_{\F} &\leq 4\Delta^2(\r-1+\deg(A))^2\norm{A}_\F \\
            \norm{\cL^2(A)}_{\F} &\leq 4\Delta^2(\r-1+\deg(A))^2\norm{A}_\F
        \end{aligned}
    \]
\end{proof}

\subsection{Proof of Theorem \ref{thm:PP-error-analysis}}
\label{sec:PP-error-proof}
\begin{proof}[Proof of Theorem \ref{thm:PP-error-analysis}]
Similarly to the fermionic case, the case where $t/\dt$ is not an integer is easily dealt with, so let us consider the case where $t/\dt$ is an integer. Denote $t_k := k\dt$ for all $k\in \{0, \dots, t/\dt \}$.
\begin{itemize}
\item[(i)] At each time step $t_k$, the degree of $A_{\MP}(t_k)$ is $\ell$. Moreover, in a similar way to Proposition \ref{prop:commutingterms}, computing $\tau^g_{\dt}(X)$ takes $O(N^{2\deg(X)})$ and $\deg \tau^g_{\dt}(X) \leq \r\deg(X)$ for all $g\in [G]$. It then takes $O(N^{2\ell}+\dots+N^{\r^{G-1}\ell}) = N^{O(\ell)}$ to compute $\trotterstep{\dt}(A_{\MP}(t_k))$, since by Proposition \ref{prop:pauli-partitionHamiltonian} we have $G\leq \r \Delta$ (i.e. $G$ is bounded by a constant). Overall, the PP algorithm then takes $\frac{t}{\dt}N^{O(\ell)}$.
\item[(ii)] This is equivalent to the proof of Theorem \ref{thm:majorana-prop}, with the exception of the new Trotter error of Theorem \ref{thm:pauli-trotter-error} which, for any $\dt \geq 0$, now satisfies
    \[
    \begin{aligned}
        \norm{\tau_{\dt}(A)-\trotterstep{\dt}(A)}_\F &\leq 2 \dt^2(\r-1+\deg A)^2 (G^2+\Delta^2) \\
        &\leq 2 (\dt\Delta)^2(\r-1+\deg A)^2 (\r^2+1)
    \end{aligned}
    \]
    since $G\leq \r\Delta$ by Proposition \ref{prop:pauli-partitionHamiltonian}.
\end{itemize}
\end{proof}

\section{The case of weakly interacting spins}
\label{appendix:spins-weak-interaction}
In the following, we turn our attention to weakly interacting systems. More specifically, we consider a $\r$-local, $\Delta$-sparse Hamiltonian where the unperturbed part acts on-site and the perturbation contains instead interactions across multiple sites. In other words, for $u\in [0,1]$ we consider
\be
\label{eq:weak-interacting-ham}
    H= H_0 + uV = \sum_{P\in\cY} h_P P + u\sum_{P\in\cZ} h_P P
\ee
with $\cY,\,\cZ$ defined as
\[
\begin{aligned}
    \cY &=\{P\in\cX\,:\,\abs{P}=1\} \,,\\
    \cZ &=\{P\in\cX\,:\,1 <\abs{P}\leq \r\} \,.
\end{aligned}
\]
Within this setting, we show the existence of a maximal time $\tmax(u)$ such that $A(t)$ can be well approximated by a low-degree operator. In particular,
\begin{theorem}
\label{thm:pauli-weak-interacting-best}
    Consider $H=H_0+uV$ to be a weakly interacting Hamiltonian as defined in \eqref{eq:weak-interacting-ham}. Let $A\in\cA$ be an arbitrary observable. Then, letting $\s=\r-1$, for all $t<\tmax(u):= \log(e/u)/(8e^2\Delta(\deg A+\s))$ we have the following upper bound
    \[
        \best_\F(H,A,t,\ell) \leq \dfrac{(t/\tmax(u))^{(\ell-\deg(A))/\s}}{1-t/\tmax(u)} \norm{A}_\F \,.
    \]
\end{theorem}
The proof of this Theorem is nearly identical to that of Theorem \ref{thm:weak-interaction}. The result is proved by showing the existence of an operator $P$ with $\deg(P) = \ell$ that satisfies
\[
    \norm{P-A(t)}_\F \leq \dfrac{(t/\tmax(u))^{(\ell-\deg(A))/\s}}{1-t/\tmax(u)} \norm{A}_\F
\]
This upper bound is satisfied by the truncation of $A(t)$ to the $k^{\text{th}}$ power of $u$. In particular, recall that the Heisenberg-evolved observable can be expanded in powers of $u$ as
\be
\label{eq:dyson-series}
    A(t) = \sum_{n=0}^\infty u^n\sum_{m_0,\dots,m_n=0}^\infty \dfrac{t^{n+\sum_i m_i}}{(n+\sum_i m_i)!}\cF^{m_n}\cG\cF^{m_{n-1}}\cdots \cF^{m_1} \cG \cF^{m_0} A \,.
\ee
with $\cF:=\ci[H_0,\cdot]$ and $\cG:=\ci[V,\cdot]$. Then, consider $P=A^\kk(t)$ to be the operator where we truncate the sum over $n$ in the equation above to the $\kk^{\text{th}}$ power. This satisfies $\deg A^\kk(t) \leq \deg A +\s\kk$ and 
\be
\label{eq:pauli-Akk-A}
    \norm{A^\kk(t)-A(t)}_\F \leq \dfrac{(t/\tmax(u))^{\kk+1}}{1-t/\tmax(u)} \norm{A}_\F
\ee
Then choose $\kk=\left\lfloor(\ell-\deg(A))/\s\right\rfloor$ to show the result. The main difference with respect to the fermionic case in showing \eqref{eq:pauli-Akk-A} is a bound on each of the operators in \eqref{eq:dyson-series}. In the following Proposition we provide a version adapted to the Pauli case.

\begin{proposition}[Frobenius norm of nested commutators]
Let $H_0$, $V$ be operators as defined in \eqref{eq:weak-interacting-ham}. Recall also $\cF:=\ci[H_0,\cdot]$ and $\cG:=[V,\cdot]$. Then, given $n\geq 1$, for any tuple $(m_0,m_1,\dots,m_n)\in\NN^{n+1}$ and $A\in\cA$, we have that
\[
    \norm{\cF^{m_n}\cG\cF^{m_{n-1}}\cdots\cF^{m_1}\cG\cF^{m_0}A}_\F \leq (4\Delta (\deg(A)+n\s))^{n+m}\norm{A}_\F
\]
with $m=\sum_{i=0}^n m_i$ and $\s=\r-1$.
\end{proposition}
\begin{proof}
    Define $B:= \cF^{m_n}\cG\cF^{m_{n-1}}\cdots\cF^{m_1}\cG\cF^{m_0}A$. Note that $\cG\cF^{m_{n-1}}\cdots\cF^{m_1}\cG\cF^{m_0}A$ has degree $d+n\s$, where $d=\deg(A)$. By Theorem \ref{thm:commutator-norm-bound}, we then have
    \be
    \label{eq:nested-comm-norm-proof}
        \norm{B}_\F \leq C(\Delta,d+n\s)^{m_n} \norm{\cG\cF^{m_{n-1}}\cdots\cF^{m_1}\cG\cF^{m_0}A}_\F
    \ee
    where $C(\Delta,k)=2\Delta\sqrt{k(k+\s)}$ for any integer $k$. Note now that the leftmost $\cG$ acts on an operator with degree $d+(n-1)\s$. Continuing from \eqref{eq:nested-comm-norm-proof}
    \[
    \begin{aligned}
        \norm{B}_\F &\leq C(\Delta,d+n\s)^{m_n} C(\Delta,d+(n-1)\s)\norm{\cF^{m_{n-1}}\cdots\cF^{m_1}\cG\cF^{m_0}A}_\F \\
        &\leq C(\Delta,d+n\s)^{m_n+1} \norm{\cF^{m_{n-1}}\cdots\cF^{m_1}\cG\cF^{m_0}A}_\F
    \end{aligned}
    \]
    Iterating the same argument, one can show
    \[
        \norm{B}_\F \leq C(\Delta,d+n\s)^{n+m} \norm{A}_\F
    \]
    We conclude the proof by noting that $C(\Delta,d+n\s) \leq 4\Delta(d+n\s)$.
\end{proof}
To conclude the appendix, using the upper bound on $\best_\F(t,\ell)$ obtained in Theorem \ref{thm:pauli-weak-interacting-best} to upper bound the error of PP in Theorem \ref{thm:PP-error-analysis}, we obtain the following corollary
\begin{corollary}
Let $H = H_0 + u V$ be a $\r$-local, $\Delta$-sparse Hamiltonian on $N$ spins where $H_0$ and $V$ are defined as in \eqref{eq:weak-interacting-ham}. Let $A \in \cA$ be an arbitrary observable. Then, for any $0 < \eps < 1$ and any $t \leq \log(e/u) / (16 e^2 \Delta (\deg A + \r-1))$, the Pauli Propagation algorithm computes an approximation $\tilde{A} \in \cA$ in complexity $N^{O(\log((1+t)/\eps))}$ such that $\|A(t) - \tilde{A}\|_{\F} \leq \eps \|A\|_{\F}$. The approximant $\tilde{A}$ is represented as a degree-$\ell$ Pauli polynomial where $\ell = \deg A + O(\log((1+t)/\eps))$.    
\end{corollary}
\begin{proof}
    The proof is equivalent to the Majorana case and uses the results of Theorem \ref{thm:pauli-weak-interacting-best} and Theorem \ref{thm:PP-error-analysis}
\end{proof}

\end{document}